\documentclass[manuscript,screen]{acmart}

\AtBeginDocument{%
  \providecommand\BibTeX{{%
    \normalfont B\kern-0.5em{\scshape i\kern-0.25em b}\kern-0.8em\TeX}}}

\usepackage{amsmath,amsfonts}
\usepackage{algorithm}
\usepackage{array}
\usepackage[caption=false,font=scriptsize,labelfont=sf,textfont=sf]{subfig}
\usepackage{textcomp}
\usepackage{stfloats}
\usepackage{url}
\usepackage{verbatim}
\usepackage{graphicx}
\usepackage{amsthm}
\usepackage{algpseudocode}
\usepackage[braket, qm]{qcircuit}
\usepackage{xcolor}
\usepackage{tikz}
\usepackage{tikz-cd}

\newtheorem{definition}{Definition}
\newtheorem{example}{Example}
\newtheorem{proposition}{Proposition}
\newtheorem{lemma}{Lemma}
\newtheorem{corollary}{Corollary}
\newtheorem{theorem}{Theorem}
\def\ket #1{\vert #1\rangle}
\def\bra #1{\langle #1\vert}
\newcommand{\ketbra}[2]{\ensuremath{\ket{#1}\!\bra{#2}}}
\newcommand{\braket}[2]{\ensuremath{\langle {#1} \vert {#2} \rangle}}

\newcommand{\tr}[0]{\textup{tr}}
\newcommand{\SPD}[0]{\textup{SPD}}
\newcommand{\CUT}[0]{U_{\text{\tiny{CUT}}}}
\begin{document}

\title{Automatic Test Pattern Generation for Robust Quantum Circuit Testing}

\author{Kean Chen}
\affiliation{
  \institution{Institute of Software, Chinese Academy of Sciences and University of Chinese Academy of Sciences, Beijing, China, and Department of Computer and Information Science, University of Pennsylvania, Philadelphia}
  \country{USA}
}
\email{keanchen.gan@gmail.com}

\author{Mingsheng Ying}
\affiliation{
  \institution{Centre for Quantum Software and Information, University of Technology Sydney}
  \country{Australia}
}
\email{Mingsheng.Ying@uts.edu.au}


\begin{abstract}
Quantum circuit testing is essential for detecting potential faults in realistic quantum devices, while the testing process itself also suffers from the inexactness and unreliability of quantum operations. This paper alleviates the issue by proposing a novel framework of automatic test pattern generation (ATPG) for robust testing of logical quantum circuits. We introduce the stabilizer projector decomposition (SPD) for representing the quantum test pattern, and construct the test application (i.e., state preparation and measurement) using Clifford-only circuits, which are rather robust and efficient as evidenced in the fault-tolerant quantum computation. However, it is generally hard to generate SPDs due to the exponentially growing number of the stabilizer projectors. To circumvent this difficulty, we develop an SPD generation algorithm, as well as several acceleration techniques which can exploit both locality and sparsity in generating SPDs. The effectiveness of our algorithms are validated by 1) theoretical guarantees under reasonable  conditions, 2) experimental results on commonly used benchmark circuits, such as Quantum Fourier Transform (QFT), Quantum Volume (QV) and Bernstein-Vazirani (BV) in IBM Qiskit. %For example, test patterns are automatically generated by our algorithm for a 10-qubit QFT circuit, and then a fault is detected by simulating the test application with detection accuracy higher than 91\%.
\end{abstract}

\begin{CCSXML}
<ccs2012>
<concept>
<concept_id>10010583.10010737.10010748</concept_id>
<concept_desc>Hardware~Test-pattern generation and fault simulation</concept_desc>
<concept_significance>500</concept_significance>
</concept>
<concept>
<concept_id>10010583.10010786.10010813.10011726</concept_id>
<concept_desc>Hardware~Quantum computation</concept_desc>
<concept_significance>500</concept_significance>
</concept>
</ccs2012>
\end{CCSXML}

\ccsdesc[500]{Hardware~Test-pattern generation and fault simulation}
\ccsdesc[500]{Hardware~Quantum computation}

\keywords{Quantum circuit, circuit testing, ATPG}


\maketitle

\section{Introduction}\label{sec-intro}

%Quantum computing has been realized in various platforms, including superconducting devices and photonic devices. The former offer high-fidelity gate implementations, while the latter benefit much from the development of classical photonic integration and has made great progress in integrating photonic quantum chips~\cite{pelucchi2022potential,wang2020integrated}. Nevertheless, scalability remain challenges for all platforms due to the imperfections of quantum operations. 
Most quantum algorithms, such as Shor’s algorithm~\cite{shor1994algorithms}, Grover’s algorithm~\cite{grover1996fast} and HHL algorithm~\cite{harrow2009quantum}, require execution on logical quantum circuits.
While the quantum error-correcting codes provide some degree of fault tolerance for logical quantum circuits,
%While the logical quantum circuits exhibit some degree of fault tolerance,
the overall error rate remains non-negligible, especially at large scales. Therefore, developing effective testing methods for logical quantum circuits is essential for advancing large-scale quantum computing.
%(particularly  photonic integrated circuits, of which the structure is more similar to classical circuits, where gates are fabricated and integrated on a chip).

\begin{figure}[t]
\centering
\includegraphics[width=0.7\linewidth]{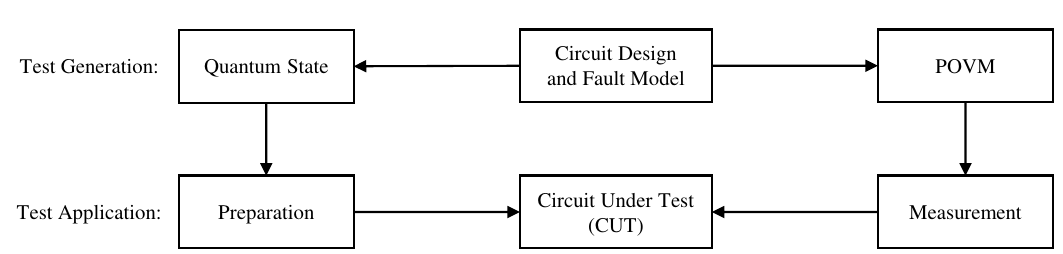}
\vspace{-0mm}
\caption{Components of quantum ATPG.}
\label{fig-fw}
\vspace{-3mm}
\end{figure}

\subsection{Challenges in Quantum ATPG}
For classical circuit testing, an automatic test pattern generation (ATPG) algorithm called the D-algorithm, was first proposed by Roth~\cite{roth1966diagnosis}, where a new logical value \(D\) is introduced to represent both the good and faulty circuit values. After that, many ATPG algorithms have been proposed and employed in industry. 
%The classical ATPG algorithms face two challenging tasks: fault-excitation and fault-propagation. The former aims to find the input patterns that can generate suitable value at the wire before the fault site, and the latter lets the output difference of the fault site be reflected at the circuit output. 
However, it is hard to directly adapt classical ATPG algorithms for quantum circuits. There are several challenges that impede quantum ATPG algorithms:
\begin{enumerate}
\item \textbf{Computational complexity} --- computation time required to generate test patterns. At first glance, one may assume that the reversibility of quantum circuits makes the fault-excitation and fault-propagation always possible and straightforward. After all, the reversibility property ensures that fault signals can traverse the circuit bidirectionally, reaching both the primary input and output of the circuit. However, this approach requires matrix multiplication with sizes exponential in the number of qubits, leading to significant inefficiency in computational complexity.
\item \textbf{Test time complexity} --- number of experiments required to obtain valid and reliable testing results. The probabilistic nature and inherent uncertainty of quantum mechanisms pose challenges in efficiently extracting information from a given quantum device. For instance, fault effects may manifest as small rotations in specific directions or subtle perturbations in the phase, resulting in low distinguishability between faulty and fault-free gates. Detecting such faults requires conducting numerous repeated experiments to obtain valid results with a sufficiently high probability (see Section \ref{discriminate-gates} and Example~\ref{example-427322} for more details). Consequently, the associated test time complexity can become prohibitive.
%More discussions are given in Section \ref{sec-131625} and Section \ref{discriminate-gates}. %While quantum state tomography~\cite{haah2017sample,gross2010quantum} can reconstruct the classical description of the quantum state, the associated test time complexity can be prohibitive. 
\item \textbf{Test equipment complexity} --- complexity of the quantum circuit required to implement the test pattern (i.e., complexity of the quantum equipment for test application, including state preparation and measurement, see Fig. \ref{fig-fw}). Examining the impact of test equipment complexity is crucial. Consider a straightforward scenario: suppose the fault-site resides at the end of a circuit, and the state capable of exciting the fault is simply \(\ket{0}\). In this case, the primary input state should be \(U^{-1}\ket{0}\), ensuring that the state-preparation circuit can trivially be \(U^{-1}\), where \(U\) denotes the unitary corresponding to the Circuit Under Test (CUT). Interestingly, both \(U^{-1}\) and \(U\) share the same circuit complexity (up to a constant scalar). This seemingly innocuous detail has far-reaching implications. The complexity of the test equipment can be remarkably high, potentially rivaling that of the CUT itself. Furthermore, due to this elevated complexity, the test equipment is susceptible to non-negligible errors --- often referred to as State Preparation and Measurement (SPAM) errors~\cite{harper2020efficient}. These errors render the test equipment itself faulty. Consequently, the complexity of the test equipment not only impacts efficiency but also significantly influences the robustness and reliability of the testing process. \label{item-5261712}
\end{enumerate}

\subsection{Related Works}
These challenges have been partially addressed or alleviated in several seminal works on quantum ATPG. Paler et al.~\cite{paler2012detection} proposed the Binary Tomographic Test (BTT) which meets the probabilistic nature of quantum circuits. A BTT is a pair of test input state and test measurement. They first simulate both fault-free and faulty quantum circuits on the BTTs. Then the CUT is executed multiple times with the same BTTs. By analyzing the output distributions and comparing them to the simulation results, faults can be detected. However, in \cite{paler2012detection}, the choice of test input states and test measurements is limited to the computational basis. While this approach simplifies test equipment complexity, it suffers from inefficiencies in test time complexity. Bera~\cite{bera2017detection} improved the BTT algorithm to account for the variety of quantum fault models. He adopts the unitary discrimination protocol~\cite{acin2001statistical} with Helstrom measurement~\cite{helstrom1969quantum}, where the input state and measurement are adaptively selected based on the specific fault being targeted. As a result, an arbitrary single fault can be detected with high probability, significantly improving test time complexity. However, it’s worth noting that Bera’s approach~\cite{bera2017detection} sacrifices test equipment complexity, making it less robust in the context of quantum circuit testing.

On the other hand, Randomized Benchmarking (RB)~\cite{emerson2005scalable,magesan2011scalable}, a widely used technique in quantum information community, has achieved great successes in assessing quantum gate. One crucial advantage of randomized benchmarking lies in its robustness against State Preparation and Measurement (SPAM) errors. However, RB relies on the gate-independent error assumption (or its weaker variant). Unfortunately, this assumption does not hold universally for logical (fault-tolerant) quantum circuits, since the implementations, and thus the fault models, of different logical gates vary greatly~\cite{fowler2012surface,zhou2000methodology}.
Also, RB requires the ability to dynamically apply different gate sequences. This approach is well-suited for certain physical implementations, such as superconducting quantum circuits, where gates are dynamic pulses sequentially applied to reused physical qubits. However, photonic devices operate differently. Their gates are fixed hardware components integrated onto a chip~\cite{pelucchi2022potential,wang2020integrated}, resembling classical circuits. Due to this fixed gate structure, the assumption of a dynamic gate sequence may not hold for photonic devices.
Moreover, faults can arise due to the integration of specific gates within a quantum circuit. These faults may not be solely dependent on an individual gate but can emerge when the gate is integrated into a circuit and interacts with other gates~\cite{harper2020efficient,harris2014efficient}. In such cases, RB is not sufficient because it does not preserve or consider the overall circuit structure that we are concerned with. Therefore, unlike RB, which primarily focuses on testing individual quantum gates, there is a need for effective testing of an entire logical quantum circuit as a whole.

\begin{figure}[t]
\centering
\subfloat[Stuck-at-0 fault]{
  \includegraphics[width=0.3\linewidth]{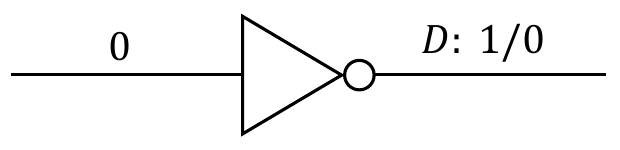}
  }\quad\quad
\subfloat[Missing gate fault]{
  \includegraphics[width=0.32\linewidth]{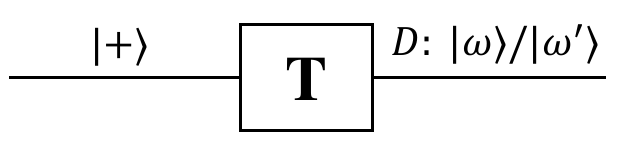}
  }
\caption{An example of the quantum analogue of classical \(D\)-value, where \(\ket{\omega}/\ket{\omega'}\) corresponds to the Helstrom measurement for distinguishing quantum states \(T\ket{+}\) and \(\ket{+}\)}
\label{fig-D}
\vspace{-3mm}
\end{figure}

\subsection{Main Idea}\label{sec-5271640}
In this paper, we propose a novel framework of quantum ATPG.
%in which the aforementioned difficulties can be handled in a principled way.
%To motivate our approach, let us 
We begin by considering an quantum analogue of the classical D-algorithm. Suppose at the fault-site, say quantum gate \(U_i\), we have a fault model \(U'_i\). By the unitary discrimination protocol~\cite{acin2001statistical} and Helstrom measurement~\cite{helstrom1969quantum}, we can obtain the input state \(\ket{\psi}\) and projective measurement \(\{\ketbra{\omega}{\omega},\ketbra{\omega'}{\omega'}\}\) for optimal distinguishability between \(U_i\) and \(U'_i\) (i.e., ensuring optimal test time complexity). This measurement is analogue to the  \(D\)-value in the classical D-algorithm. Indeed, it can be seen as a composite ``logical value'' \(\ket{\omega}/\ket{\omega'}\) where the first entry \(\ket{\omega}\) represents the ``value'' of fault-free gate and the second entry \(\ket{\omega'}\) represents the ``value'' of faulty gate. 
%Note that the ``value'' here does not indicate the exact quantum state but the optimal measurement for distinguishing the fault-free and faulty states. This is because in an operational view of quantum theory, the fault-effect is no more than the optimal measurement for this specified fault. 
A simple example is shown in Fig. \ref{fig-D}. 
%where \(\ket{+}\) is the optimal input state for distinguishing the \(T\) gate and the identity gate (\textit{i.e.}, missing gate), \(\ket{\omega}/\ket{\omega'}\) corresponds to the Helstrom measurement for distinguishing quantum states \(T\ket{+}\) and \(\ket{+}\). 
Then, similar to the D-algorithm, we propagate \(\ket{\psi}\) and \(\ket{\omega}/\ket{\omega'}\) to the primary input and output of the circuit to obtain the test pattern. 
%The most straightforward approach for propagation is through direct matrix multiplication. However, this method involves matrices with sizes that grow exponentially with the number of qubits, leading to inefficiencies in computational complexity. 
%To address this challenge, we need to explore more acceleration techniques, which will be discussed in Section \ref{sec-pattern} and Section \ref{sec-213336}.

\begin{figure}[t]
    \centering
    \includegraphics[width=0.7\linewidth]{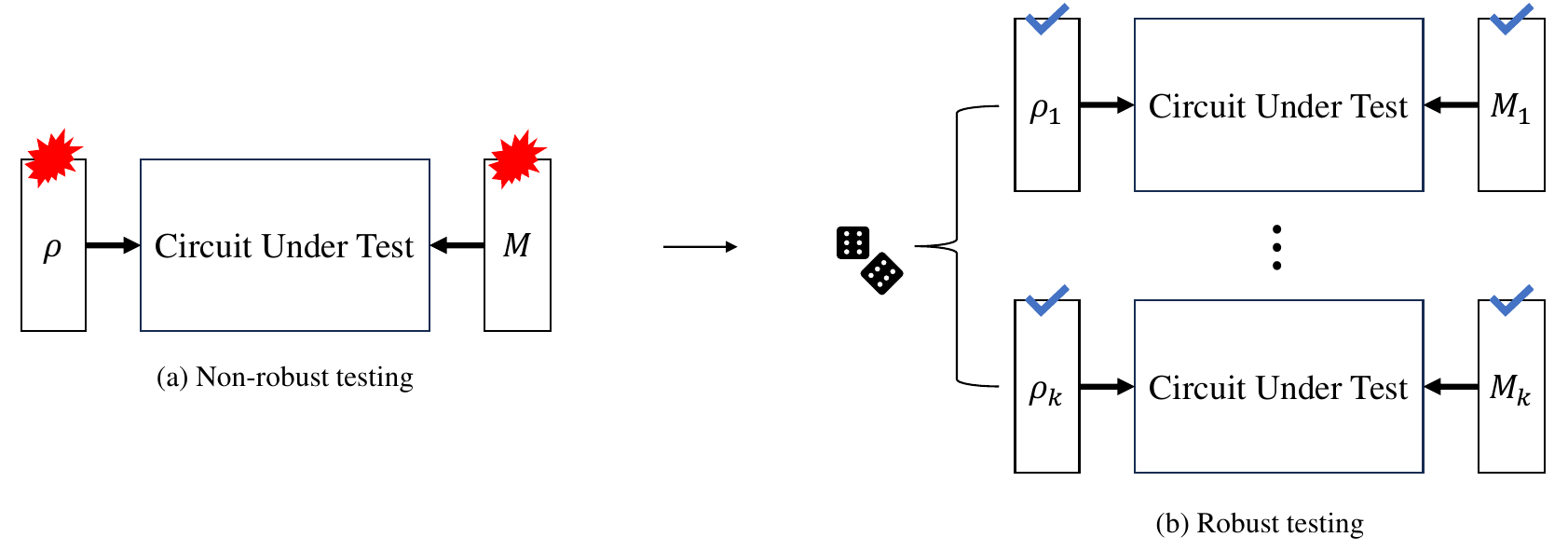}
    \caption{(a) Non-robust testing: the test pattern \((\rho,M)\) is directly implemented, which is susceptible to SPAM errors (see the challenge (\ref{item-5261712})). (b) Robust testing: the test pattern is decomposed into simpler ones (i.e., \((\rho_i,M_i)\)), which are implemented using Clifford-only circuits, and are then combined through classical randomness.}
    \label{fig-5261603}
\end{figure}

However, as previously mentioned, a critical challenge lies in the test equipment complexity and the necessity for robustness in test application (see the above discussion on challenge (3)). To address this issue, we decompose the high-complexity test pattern into a set of simpler test patterns ------ each simpler test pattern is implemented via Clifford-only circuits, and these test patterns are combined through classical randomness, as shown in Fig.~\ref{fig-5261603}. For example, the direct implementation of the test patterns on a QFT\_10 circuit leads to an average circuit size (i.e., the test equipment complexity) of 247, while our method leads to an average circuit size of only 29 (more details can be found in Table~\ref{table-526039}). Moreover, our method only uses Clifford circuits, which exhibit significant advantages in terms of both efficiency and robustness, as shown below:
\begin{enumerate} 
\item[(i)] The Clifford gates are nearly ideal~\cite{bravyi2005universal} in the logical (fault-tolerant) quantum circuits based on stabilizer codes, since the Clifford gates have rather robust and cheap implementation~\cite{fowler2009high,zhou2000methodology}.
\item[(ii)] Any Clifford operation can be efficiently synthesized to a relatively small Clifford circuit with size no more than \(O(n^2/\log n)\)~\cite{aaronson2004improved,bravyi2021clifford}, where \(n\) is the number of qubits. 
\item[(iii)] The Clifford-universal gate set (\textit{e.g.} Hadamard, Phase and controlled-NOT) is finite and discrete. Thus the Clifford gates can be carefully calibrated on specific quantum devices. This contrasts with the continuous rotation gates such as those in the QFT circuit.
\end{enumerate}
The above advantages ensure that using the Clifford circuits for test application is simpler and more robust against SPAM error than using general quantum circuits. It should be noted that our robust testing method introduces an additional overhead on the number of experiments. This can be viewed as a trade off between quantum resources and classical resources. Specifically, we transform the cost of quantum resources (i.e., complexity of quantum test equipment) to the cost of classical resources (i.e., repeated experiments sampled from classical randomness). This trade off is valuable in quantum circuit testing as the classical resources are more available and reliable for the near-term devices.

Based on this observation, we introduce the stabilizer projector decomposition (SPD) to represent the test pattern. By applying a sampling algorithm on the SPD, the test application can be constructed using Clifford-only circuits with classical randomness. 
%It is worth noting that here we use classical randomness rather than quantum randomness (\textit{i.e.}, superposition) because preparing an arbitrary superposition state is much harder than sampling from the corresponding classical distribution. 
The number of required experiments in our sampling algorithm is explicitly related to the 1-norm metric of SPD.
The optimal SPD (\textit{i.e.}, that with  minimal 1-norm) for given test pattern can be obtained by solving a convex optimization problem. But in general, it is practically intractable to solve this optimization problem due to the exponentially growing of the problem size. To mitigate this issue, we propose an SPD generation algorithm with several acceleration techniques, where both locality and sparsity are exploited in the SPD calculation. It is proved that the locality exploiting algorithm is optimal in the sense that it reveals the most locality of given SPD, under some mild conditions. To demonstrate the effectiveness and efficiency of our approach, the overall quantum ATPG framework and the proposed SPD generation algorithm are implemented and evaluated on several commonly used benchmark circuits in IBM Qiskit. The code is available at Github\footnote{\url{https://github.com/cccorn/Q-ATPG}}.

\subsection{Organization}
\begin{figure}
    \centering
    \includegraphics[width=0.7\linewidth]{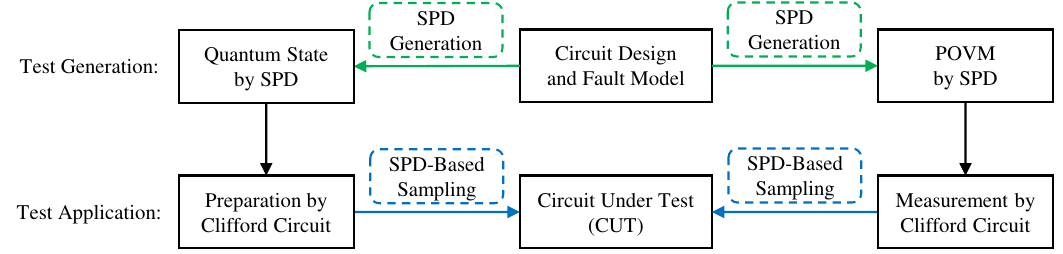}
    \caption{Our quantum ATPG framework. Linked to Fig.~\ref{fig-fw}, the quantum state and POVM are represented by SPDs, and the state preparation and measurement are implemented using Clifford circuits. The test application and test generation are conducted by the SPD-based sampling algorithm (see Section~\ref{sec-testapp}) and the SPD generation algorithm (see Section~\ref{sec-spdgen}), respectively.}
    \label{fig-5281343}
\end{figure}
Our quantum ATPG framework is visualized in Fig.~\ref{fig-5281343}. 
The paper is then organized as follows.
In Section~\ref{sec-intro} (i.e. this section), we discussed 3 main challenges in quantum circuit testing, i.e., computational complexity, test time complexity and test equipment complexity, followed by our main idea to handle these challenges. 
In Section~\ref{Preliminaries}, we
provide the necessary background knowledge, and 
further explain the challenges with a series of detailed examples, i.e., Example~\ref{example-151418}~--~\ref{example-527336}.
In Section~\ref{sec-testapp}, we present an SPD-based sampling algorithm (cf. the blue dashed box in Fig.~\ref{fig-5281343}), targeting on test equipment complexity. 
In Section~\ref{sec-spdgen}, we present an efficient SPD generation algorithm (cf. the green dashed box in Fig.~\ref{fig-5281343}), targeting on both test time complexity and computational complexity.
In Section~\ref{sec-eva}, we experimentally evaluate the proposed method on different benchmark circuits.
%The overview of our ATPG framework is shown in Fig.~\ref{fig-5281343}. 

%The paper is organized as follows.
%In Section \ref{Preliminaries},  the main technical tools needed in this paper are reviewed, including quantum circuits and the fault models,  unitary discrimination protocol and the corresponding test pattern, stabilizer formalism and Clifford gates. 
%In Section \ref{sec-testapp}, the notion of stabilizer projector decomposition (SPD) with a sampling algorithm is proposed to implement the robust quantum test application. 
%In Section \ref{sec-spdgen}, an SPD generation algorithm is developed, with several practical acceleration techniques and certain theoretical guarantees.
%In Section \ref{sec-eva}, the proposed algorithms are implemented and evaluated on several commonly used benchmark circuits in IBM Qiskit.
%In Section \ref{sec-con}, we summarize the contributions of this work and provide an outlook for future work.
%Due to the limited space, the proofs of all lemmas, propositions and theorems are omitted in the main text but given in the Appendix. 

\section{Preliminaries}\label{Preliminaries}
%\subsection{Quantum Circuits and Fault Models}\label{sec-131625}
\subsection{Quantum Circuits}
Quantum circuits are made up of qubits (quantum bits), quantum gates and measurements.
Using the Dirac notation, a (pure) state of a single qubit is represented by \(\ket{\psi}=\alpha_0\ket{0}+\alpha_1\ket{1}\) with complex numbers \(\alpha_0\) and \(\alpha_1\) satisfying \(|\alpha_0|^2+|\alpha_1|^2=1\). 
More generally, a state of $n$ qubits can be written as a superposition of $n$-bit strings: 
\begin{equation}
\ket{\psi}=\sum_{x\in \{0,1\}^n}\alpha_x\ket{x}
\end{equation}
where complex numbers \(\alpha_x\) satisfies the normalization condition \(\sum_{x} |\alpha_x|^2=1\). 
If we identify a string \(x=x_1... x_n\in \{0,1\}^n\) with the integer   \(x=\sum_{i=1}^nx_i\cdot 2^{i-1}\), then this state can also be represented by the \(2^n\)-dimensional column vector \(\ket{\psi}=(\alpha_0,\ldots,\alpha_{2^n-1})^T\). In addition, we use \(\bra{\psi}\) to denote the row vector \((\alpha_0^*,\ldots,\alpha_{2^n-1}^*)\), where the symbol \(*\) represents the complex conjugate. Given an ensemble of pure states \(\{(p_i,\ket{\psi_i})\}\) where \(\sum_i p_i=1\), and suppose that the quantum system is in one of the \(\ket{\psi_i}\) with probability \(p_i\). We then refer to this system as being in a mixed quantum state. Such a state can be conveniently represented using a mathematical tool called the density operator \(\rho=\sum_i p_i\ketbra{\psi_i}{\psi_i}\). When the context is clear, the term ``quantum state'' will be used to mean either a pure state or a mixed state.

A quantum gate on $n$ qubits is mathematically described by a \(2^n\times 2^n\) unitary matrix \(U=\{u_{ij}\}\) satisfying \(U^\dag U=I_n\), where \(U^\dag\) stand for the conjugate transpose of \(U\) and \(I_n\) the identity matrix of dimension \(2^n\) (we may omit the subscript of \(I_n\) when it does not cause confusion). If a pure state \(\ket{\psi}\) (or mixed state \(\rho\)) is input to this gate, then the output from the gate is represented by the pure state \(U\ket{\psi}\) (or mixed state \(U\rho U^\dag\)). 
Some commonly used quantum gates are shown in Fig. \ref{fig-ex-qgate}.
Then, a quantum circuit is a sequence of quantum gates: \((U_1,\ldots ,U_d)\), where \(d\geq 1\), and \(U_i\,\, (1\leq i\leq d)\) are quantum gates. 
We will use \(U_{i:j}\,\, (i\leq j)\) to denote the unitary matrix  corresponding to the circuit slice from \(i\)-th gate to \(j\)-th gate; that is, \(U_{i:j}=U_j \ldots U_{i+1} U_{i}\). In this work, we assume that the elementary gate set under consideration is composed of Clifford gates and Pauli rotation gates (i.e., \(e^{i\theta P}\) where \(P\) is a Pauli operator and \(\theta \in [-\pi,\pi]\), more details can be found in Section~\ref{sec-stabilizer}). Detailed parameters of the benchmark circuits used in this paper are summarized in Table~\ref{table-bench}.

To extract classical information from a quantum state, we need to apply measurement on it. Mathematically, a quantum measurement can be described by a set of positive operators \(M=\{M_i\}\), where \(\sum_i M_i=I\). This set \(M\) is known as the Positive Operator-Valued Measure (POVM). When we perform a measurement \(M\) on a quantum state \(\rho\), we observe an outcome \(i\) with probability \(\tr(M_i\rho)\). 
Unlike the classical systems, the quantum states are generally disturbed after measurement, making it challenging to reuse the same quantum state for multiple measurements.
\begin{example}\label{example-151418}
Consider the measurement on a qubit in the computational basis: \(\{M_0,M_1\}\), where \(M_0=\ketbra{\omega}{\omega}\), \(M_1=\ketbra{\omega'}{\omega'}\), in which \(\ket{\omega}=\frac{1}{\sqrt{2}}e^{-i5\pi/16}\ket{0}+\frac{1}{\sqrt{2}}e^{i5\pi/16}\ket{1}\) and \(\ket{\omega'}=\frac{1}{\sqrt{2}}e^{i3\pi/16}\ket{0}+\frac{1}{\sqrt{2}}e^{-i3\pi/16}\ket{1}\). If we perform the measurement \(\{M_0,M_1\}\) on a qubit in state \(\ketbra{+}{+}\), where \(\ket{+}=\frac{1}{\sqrt{2}}\ket{0}+\frac{1}{\sqrt{2}}\ket{1}\), then we get outcome \(0\) with probability \(\tr(M_0\ketbra{+}{+})=\cos^2(5\pi/16)\), and get outcome \(1\) with probability \(\tr(M_1\ketbra{+}{+})=\sin^2(5\pi/16)\).
\end{example}

\begin{figure}
\centering
\hspace{-2.8mm}\subfloat{
\raisebox{-1mm}{Hadamard:\quad\,\,}
\scalebox{1.0}{
\Qcircuit @C=1.0em @R=0.65em @!R {
	 	\nghost{} & \qw & \gate{\mathrm{H}} &\qw & \qw& \nghost{}
}}
\raisebox{-1mm}{\(\frac{1}{\sqrt{2}}\begin{bmatrix}1&1\\1&-1\end{bmatrix}\)}
}
\par
\vspace{3mm}
\hspace{-8.1mm}\subfloat{
\raisebox{-1mm}{Pauli-X (NOT):\,\,\,}
\scalebox{1.0}{
\Qcircuit @C=1.0em @R=0.65em @!R {
	 	\nghost{} & \qw & \gate{\mathrm{X}} &\qw & \qw& \nghost{}
}}
\raisebox{-1mm}{\(\,\quad\begin{bmatrix}0&1\\1&0\end{bmatrix}\)}
}
\par
\vspace{3mm}
\hspace{-3.9mm}\subfloat{
\raisebox{-3mm}{Controlled-NOT:\,\,}
\scalebox{1.0}{
\Qcircuit @C=1.0em @R=0.65em @!R {
	 	\nghost{} & \qw & \ctrl{1} &\qw & \qw& \nghost{}\\
	 	\nghost{} & \qw & \targ &\qw & \qw& \nghost{}
}}
\raisebox{-3mm}{\(\begin{bmatrix}1&0&0&0\\0&1&0&0\\0&0&0&1\\0&0&1&0\end{bmatrix}\)}
}
\caption{Examples of single and multi-qubit quantum gates.}
\label{fig-ex-qgate}
\end{figure}

%\begin{example}[Quantum Gates]
%Some commonly used quantum gates are shown in Fig. \ref{fig-ex-qgate}, including:
%\begin{enumerate}
%\item Single-qubit quantum gates such as the Hadamard gate and Pauli-X (NOT) gate.
%\item Multi-qubit quantum gates such as the controlled-NOT (CNOT) gate.
%\end{enumerate}
%\end{example}

\subsection{Fault Models}\label{sec-5231621}
Let \(\CUT:=(\tilde{U}_1,\ldots ,\tilde{U}_d)\) be the circuit under test (CUT). Each gate \(\tilde{U}_i\) is specified by its fault-free version \(U_i\) and has the potential to be faulty. We assume that its fault model \(U'_i\) is known in advance. Our goal is to determine whether the given CUT is fault-free or faulty. 
In this work, we adopt \textit{the single-fault assumption} which is typically used in test generation and evaluation for both classical circuits~\cite{wang2006vlsi} and quantum circuits \cite{paler2012detection,bera2017detection}: the cause of a circuit failure is attributed to only one faulty gate. This assumption is reasonable for logical (fault-tolerant) quantum circuits, since the effect of noise on single logical gate is fairly small such that the probability of single fault occurring dominates the probability of multiple faults occurring~\cite{harper2019fault,postler2022demonstration}.
\begin{example}\label{example-151419}
Consider a simple example of the fault models. In the Quantum Fourier Transform (QFT) circuit (see Fig.~\ref{fig-152206}), the middle \(Z\)-axis-\(\pi/4\)-rotation gate has the potential to suffer from a missing gate fault. That is, the fault-free gate is \(U=R_Z(\frac{\pi}{4})=\begin{bmatrix}e^{-i\pi/8}&0\\0&e^{i\pi/8}\end{bmatrix}\), and its fault model is \(U'=I=\begin{bmatrix}1&0\\0&1\end{bmatrix}\).
\end{example}

\subsection{Discrimination of Quantum Gates}\label{discriminate-gates}

A basic step in quantum ATPG is to discriminate a gate $U_i$ from its faulty version $U_i^\prime$. This problem has been well studied by the quantum information community \cite{acin2001statistical,helstrom1969quantum,holevo1973statistical}.
A general manner is by inputting a mixed state \(\rho\) to this unknown gate and then performing the measurement to obtain the result. Mathematically, the distinguishability between \(U_i\rho U_i^\dag, U'_i\rho {U'}_i^{\dag}\), \textit{i.e.}, the trace distance \(tr(|U_i\rho U_i^\dag-U'_i\rho {U'}_i^{\dag}|)\) by the Holevo-Helstrom theorem~\cite{helstrom1969quantum,holevo1973statistical}, must be maximized. Due to the convexity of trace distance, the optimal input can always be achieved by a pure state \(\rho=\ketbra{\psi}{\psi}\). Thus the problem is equivalent to minimizing the inner product of \(U_i\ket{\psi}\) and \(U'_i\ket{\psi}\):
\begin{equation}
\min_{\ket{\psi}} \big|\bra{\psi}U_i^\dag U'_i\ket{\psi}\big|.
\end{equation}
Let the spectral decomposition of \(U_i^\dag U'_i\) be \(\sum_i\lambda_i\ketbra{v_i}{v_i}\), and \(\ket{\psi}=\sum_i \alpha_i \ket{v_i}\). Then the optimization problem becomes:
\begin{equation}\label{opt_input}
\min_{\ket{\psi}} \bigg|\sum_i \lambda_i \braket{\psi}{v_i} \braket{v_i}{\psi}\bigg|=\min_{\sum |\alpha_i|^2=1}\bigg|\sum_i \lambda_i |\alpha_i|^2\bigg|,
\end{equation}
which can be solved by quadratic programming.

After obtaining the optimal input state \(\ket{\psi}\) from Equation (\ref{opt_input}), the optimal measurement for distinguishing between the pair \(\ket{\phi}\coloneqq U_i\ket{\psi},\ \ket{\phi'}\coloneqq U_i'\ket{\psi}\) of output states of the two gates can be achieved by their Helstrom measurement. More specifically, assume:
\begin{equation}\label{eq-522315}
\braket{\phi}{\phi'}=re^{i\theta},
\end{equation} 
and define:
\begin{equation}\label{opt_meas}
\begin{split}
&\ket{\omega}=\frac{1}{\sqrt{2}}\bigg(\frac{\ket{\phi}+e^{-i\theta}\ket{\phi'}}{\|\ket{\phi}+e^{-i\theta}\ket{\phi'}\|}+\frac{\ket{\phi}-e^{-i\theta}\ket{\phi'}}{\|\ket{\phi}-e^{-i\theta}\ket{\phi'}\|}\bigg),\\
&\ket{\omega'}=\frac{1}{\sqrt{2}}\bigg(\frac{\ket{\phi}+e^{-i\theta}\ket{\phi'}}{\|\ket{\phi}+e^{-i\theta}\ket{\phi'}\|}-\frac{\ket{\phi}-e^{-i\theta}\ket{\phi'}}{\|\ket{\phi}-e^{-i\theta}\ket{\phi'}\|}\bigg).
\end{split}
\end{equation}
Then the Helstrom measurement is the POVM with measurement operators \(\ketbra{\omega}{\omega}, \ketbra{\omega'}{\omega'}\), corresponding to the predictions ``fault-free'' and ``faulty'' respectively. The optimal success probability for distinguishing between faulty and fault-free gates in a single experiment is given by 
\begin{equation}\label{eq-4272207}
\frac{1}{2}+\frac{1}{2}\sqrt{1-r^2}\in (\frac{1}{2},1],
\end{equation} 
where \(r\) is defined in Equation \eqref{eq-522315}.
Note that when the fault effect is small (e.g., subtle perturbations on the phase), \(r\) will be close to \(1\) so that the optimal strategy will not be much better than a random guess. In such cases, conducting a significantly larger number of repeated experiments becomes necessary to reliably detect the fault.
%Moreover, if we define the order $A\sqsubseteq B$ between two matrices $A$ and $B$ by requiring that $B-A$ is positive semi-definite, then any POVM measurement with measurement operators \(\{M,I-M\}\) satisfying \(M\sqsupseteq \ketbra{\omega}{\omega}, I-M\sqsupseteq \ketbra{\omega'}{\omega'}\) is also optimal for discriminating states \(\ket{\phi}\) and \(\ket{\phi'}\).
\begin{example}\label{example-160307}
Consider the fault model given in Example~\ref{example-151419}. The optimal input state distinguishing between the faulty and fault-free gates, by solving optimization problem (\ref{opt_input}), is \(\ket{\psi}=\ket{+}=\frac{1}{\sqrt{2}}\ket{0}+\frac{1}{\sqrt{2}}\ket{1}\). The outputs of fault-free and faulty gates on \(\ket{\psi}\) are \(\frac{1}{\sqrt{2}}(e^{-i\pi/8}\ket{0}+e^{i\pi/8}\ket{1})\) and \(\ket{+}\), respectively. Then, the Helstrom measurement, given by Equation (\ref{opt_meas}), is exactly the measurement presented in Example~\ref{example-151418}, where \(M_0,M_1\) correspond to ``fault-free'' and ``faulty'', respectively. Therefore, when the gate is faulty, our strategy will output ``faulty'' with probability \(\tr(M_1\ketbra{+}{+})=\sin^2(5\pi/16)\) in a single experiment, as shown in Example~\ref{example-151418}. In addition, if the gate is fault-free, simple calculation shows that our strategy will output ``fault-free'' also with probability \(\sin^2(5\pi/16)\). Indeed, \(\sin^2(5\pi/16)\) is the optimal probability of successfully distinguishing these gates in a single experiment.
\end{example}

\begin{example}\label{example-427322}
In Example~\ref{example-160307}, we have shown that the success probability for detecting the missing gate fault on \(R_Z(\frac{\pi}{4})\) gate in a single experiment is \(\sin^2(5\pi/16)\). To obtain an overall success probability \(\geq 0.9\), we need to repeat the experiment for \(11\) times (this can be seen by calculating the cumulative binomial probability). In fact, the test time complexity (number of repeated experiments) varies a lot on different gates, for example:
\begin{itemize}
\item Detecting the missing gate fault on \(H\) gate with success probability \(\geq 0.9\) requires only \(1\) experiment,
\item Detecting the missing gate fault on \(R_Z(\frac{\pi}{16})\) gate with success probability \(\geq 0.9\) requires \(171\) repeated experiments.
\end{itemize}
In general, to amplify the success probability from \(\frac{1}{2}+\frac{1}{2}\sqrt{1-r^2}\) (see Equation (\ref{eq-4272207})) to \(0.9\), the test time complexity (number of repeated experiments) is \(\Theta(\frac{1}{1-r^2})\).
\end{example}

\subsection{Test Patterns for Single-Fault Quantum Circuits}\label{sec-pattern}
The above discussion considers gate-level discrimination. Now we move on to consider the circuit-level discrimination between the faulty and fault-free circuits. It is impractical to solve the optimization problem (\ref{opt_input}) at the full-circuit level due to the high computational complexity. Specifically, it requires \(O(d\cdot N^3)\) arithmetic operations to first compute the matrix representing the whole circuit, and then \(O(N^3)\) arithmetic operations to solve the optimization problem (\ref{opt_input}), where \(d\) is the number of gates and \(N=2^n\) is the dimension of the whole quantum system.
Nevertheless, since there is at most one faulty gate in the circuit (single-fault assumption, see Section \ref{sec-5231621}), we can locally solve the discrimination problem for the faulty and fault-free gates at the fault site, which generally act on a small sub-system (say, a single qubit or two qubits), obtaining the results \(\ket{\psi},\ket{\omega},\ket{\omega'}\) as in Equations (\ref{opt_input}) and (\ref{opt_meas}). Then we propagate the fault excitation signal \(\ket{\psi}\) and fault effect \(\ket{\omega}/\ket{\omega'}\) to the beginning and end of the circuit respectively. Formally, we have: 
\begin{proposition}\label{prop-propagate}
If quantum state \(\rho\) and measurement \(\{M,I-M\}\) satisfy the following conditions: 
\begin{equation}\label{opt-cio}
\begin{gathered}
\rho \sqsubseteq U_{1:i-1}^\dag\,(Q \otimes \ketbra{\psi}{\psi}) \, U_{1:i-1},\\
U_{i+1:d} \,(Q \otimes \ketbra{\omega}{\omega}) \,U^\dag_{i+1:d} \sqsubseteq M,\\
U_{i+1:d} \,(Q \otimes \ketbra{\omega'}{\omega'}) \,U^\dag_{i+1:d} \sqsubseteq I-M,
\end{gathered}
\end{equation}
where \(\sqsubseteq\) is the Loewner order and $Q$ is some projection on the state space of the qubits that are not in gates $U_i$ and $U'_i$, 
then \(\rho\) and \(\{M,I-M\}\) are optimal state and measurement at the primary input and output for distinguishing between faulty and fault-free circuits, i.e., optimal in test time complexity.
\end{proposition}

We call the pair $(\rho,M)$ satisfying constraints (\ref{opt-cio}) a test pattern of CUT $C$ with respect to its faulty version $C_i$. Note that unlike the optimization problem in previous section, there is flexibility in defining the optimal test pattern. This is because the parts of system other than the fault site can be arbitrary without affecting optimality.
However, for simplicity, we can choose the test pattern \((\rho,M)\) to be  
\begin{equation}\label{exp-rhoM}
\begin{gathered}
\rho=U^{\dag}_{1:i-1}\big(\frac{Q}{\tr(Q)}\otimes\ketbra{\psi}{\psi}\big)U_{1:i-1}\\
M=U_{i+1:d}\big(Q\otimes \ketbra{\omega}{\omega}\big) U^\dag_{i+1:d}.
\end{gathered}
\end{equation}
Note that Equation \eqref{exp-rhoM} also involves matrices of exponential sizes, which is hard to directly evaluate. More specifically, if we choose \(Q=\ketbra{\phi}{\phi}\), the computational complexity is \(O(d\cdot N^2)\) (since it only requires matrix-vector multiplications), which is better than the full-circuit level optimization but is still inefficient. This issue is further alleviated in Sections \ref{sec-213336} and \ref{sec-131311}, where several acceleration techniques are proposed. There, we use stabilizer projector decomposition (SPD) to represent the the test pattern \((\ketbra{\psi}{\psi},\ketbra{\omega}{\omega})\), and sequentially propagate the SPD to the primary input and output of the circuit.
%In particular, \(Q=I\) (the identity operator on the state space of the qubits that are not in $U_i$ and $U'_i$) is a practical choice for computational efficiency. This is because $I$ is invariant under any gate $U_j$ (i.e., \(U_jIU_j^\dag=I\)), simplifying many computations.
\begin{example}\label{example-527336}
Consider the fault model given in Example~\ref{example-151419}. The optimal state \(\ket{\psi}\) and measurement \(\{\ketbra{\omega}{\omega},\ketbra{\omega'}{\omega'}\}\) for distinguishing the fault-free and faulty gates are given in Example~\ref{example-160307} and Example~\ref{example-151418}. Then, after applying Equation (\ref{exp-rhoM}), we obtain a test pattern \((\rho,M)\), where \(\rho=\frac{1}{4}\cdot I\otimes \ketbra{0}{0}\otimes I\) and \(M=I\otimes \begin{bmatrix}0.5&0&a&b\\0&0.5&b&a\\a^*&b^*&0.5&0\\b^*&a^*&0&0.5\end{bmatrix}\), in which \(a\approx-0.135-0.326i,b\approx 0.326-0.135i\). Note that \(M\) is a intricate measurement, which lacks efficient and robust implementations. Our proposed method can decompose this challenging test pattern into several simpler test patterns, which can be efficiently implemented using the robust Clifford-only circuits. The produced results are shown in Fig.~\ref{fig-exp1}.
\end{example}

\subsection{Stabilizer Formalism and Clifford Circuits}\label{sec-stabilizer}
As discussed in Section \ref{sec-intro}, one major challenge lies in the test equipment complexity and the robustness of test application. More specifically, the difficulty arises from the hardness of implementing the state \(\rho\) and measurement \(M\) in Equation (\ref{exp-rhoM}). Our solution to this problem leverages the stabilizer formalism and the associated properties of Clifford circuits, which are commonly employed in fault-tolerant quantum computing. For the reader’s convenience, we provide a brief review of these concepts in the following subsection.

\subsubsection{Pauli Operators} Pauli operators are widely used in quantum computation and quantum information~\cite{nielsen2002quantum}. Single-qubit Pauli operators consist of \(I,X,Y,Z\) where
\begin{equation}
X=\begin{bmatrix}0&1\\1&0\end{bmatrix},\, Y=\begin{bmatrix}0&-i\\i&0\end{bmatrix},\, Z=\begin{bmatrix}1&0\\0&-1\end{bmatrix}
\end{equation}
Any \(n\)-qubit Pauli operator \(P\) is of the form:
\(P_1\otimes P_2 \otimes \cdots \otimes P_n\)
where each \(P_i\) is chosen from \(I,X,Y,Z\).
We will use \(X_i\) (or \(Y_i,Z_i\)) to denote the Pauli operator \(P_1\otimes \cdots \otimes P_n\) with \(P_i=X\) (or \(Y,Z\)) and \(P_j=I\) for all \(j\neq i\).
Suppose \(P\) is a Pauli operator, we will use \(\underline{e^{i\theta}P}\) to denote the Pauli operator by simply discarding the phase \(e^{i\theta}\), (\textit{e.g.} \(\underline{ZX}=\underline{iY}=Y\)).
We call \(e^{i\theta}P\) a \textit{signed} Pauli operator, if \(\theta\in \{0,\pi\}\) and \(P\) is a Pauli operator. Without causing confusion, the same letters (\textit{e.g.} \(P,Q\)) may be used to denote either Pauli or signed Pauli operators.
Note that any two (signed) Pauli operators \(P,Q\) either commute (\textit{i.e.}, \(PQ=QP\)) or anti-commute (\textit{i.e.}, \(PQ=-QP\)).

There is a very useful representation for the Pauli operators, using the vectors over \(\textup{GF}(2)\). In this paper, we use \(00,01,10,11\) to represent \(I,Z,X,Y\) respectively. The Pauli operator \(P=P_1\otimes \ldots\otimes P_n\) is then represented by the vector \(v(P)=(v_1,\ldots,v_{2n})\), where \(v_{2i-1},v_{2i}\) corresponds to the representation of \(P_i\). Note that the vector addition coincides with the Pauli operator multiplication up to a global phase. That is, if \(\underline{PQ}=R\) where \(P,Q,R\) are Pauli operators, then \(v(P)+v(Q)=v(R)\). 
For a signed Pauli operator \(P\), we simply define its vector representation as \(v(\underline{P})\).
We say that a set of (signed) Pauli operators are \textit{independent} if their vector representations are independent.

\subsubsection{Stabilizer Formalism} We say that a  state \(\ket{\psi}\) is \textit{stabilized} by a  signed Pauli operator \(P\) if \(P\ket{\psi}=\ket{\psi}\). In this case, \(P\) is called the \textit{stabilizer} of \(\ket{\psi}\). The basic idea of the stabilizer formalism is that many quantum states can be more easily described by their stabilizers~\cite{nielsen2002quantum}.
%\begin{example}[Pauli Stabilizer]
%Consider the EPR state 
%\begin{equation}
%\ket{\psi}=\frac{\ket{00}+\ket{11}}{\sqrt{2}}
%\end{equation}
%We can verify that \(\ket{\psi}\) is the unique quantum state (up to global phase) stabilized by Pauli operators \(X_1X_2\) and \(Z_1Z_2\). Thus \(\ket{\psi}\) can be described by its stabilizers \(\{X_1X_2, Z_1Z_2\}\).
%\end{example}
In general, suppose \(G\) is a group of \(n\)-qubit signed Pauli operators and \(-I \notin G\). We call \(G\) the \textit{stabilizer group}. Let \(V\) be the set of all \(n\)-qubit states which are stabilized by every element of \(G\). Then \(V\) forms a vector space. The projector \(S\) onto \(V\) is called the \textit{stabilizer projector} corresponding to group \(G\). When \(V\) is one-dimensional, its projector \(S\) is of the form \(S=\ketbra{\phi}{\phi}\) and \(\ket{\phi}\) is called the \textit{stabilizer state} corresponding to \(G\). 
A set of elements \(P_1,\ldots,P_l\) is said to be the \textit{generators} of \(G\) if every element of \(G\) can be written as a product of elements from the list \(P_1,\ldots,P_l\), and we write \(G=\langle P_1,\ldots,P_l\rangle\).
There are some useful properties of the stabilizer formalism. First, the elements of \(G\) commute, \textit{i.e.}, for any \(P,Q\in G\), \(PQ=QP\). 
Second, the stabilizer projector \(S\) can be represented by 
\begin{equation}
S=\prod_{i=1}^l \frac{I+P_i}{2}=\frac{1}{|G|}\sum_{P_i\in G} P_i
\end{equation}
where \(P_1,\ldots,P_l\) are the generators of \(G\).

\subsubsection{Clifford Operation}
Suppose we apply a unitary \(U\) to the stabilizer space \(V\). Let \(\ket{\phi}\in V\). Then  for any \(g\in G\): 
\(U\ket{\phi}=Ug\ket{\phi}=UgU^\dag U\ket{\phi}\)
which means that  \(U\ket{\phi}\) is stabilized by \(UgU^\dag\). Thus subspace \(U V\) is stabilized by the group \(UGU^\dag\equiv\{UgU^\dag|g\in G\}\) and the corresponding stabilizer projector is \(USU^\dag\). For a unitary \(U\), if it takes signed Pauli operators to signed Pauli operators under conjugation, \textit{i.e.}, \(UPU^\dag=P'\), then we call it a \textit{Clifford operation}. This transformation takes on a particularly appealing form. That is, for any stabilizer projector \(S\), \(USU^\dag=S'\) where \(S'\) is also a stabilizer projector. It was proved that any \(n\)-qubit Clifford operation can be implemented using Hadamard, Phase and CNOT gates, with the circuit size no more than \(O(n^2/\log n)\)~\cite{aaronson2004improved}.
%\begin{example}[Clifford Gate]
%Hadamard gate is Clifford:
%\begin{equation}
%HXH^\dag=Z,\,\, HYH^\dag=-Y,\,\, HZH^\dag=X
%\end{equation}
%\end{example}
Clifford circuits can be efficiently simulated on a classical computer according to the Gottesman-Knill theorem~\cite{gottesman1998heisenberg}. 
Furthermore, using fault-tolerant techniques  based on CSS codes, including the popular Steane code and surface codes, they can be efficiently implemented with little error rate~\cite{fowler2009high,zhou2000methodology}. This contrasts with the non-Clifford gates (\textit{e.g.} \(T\) gate), which incur more than a hundred times cost for fault-tolerant implementation. The following facts about Clifford operations will be needed in this paper:
\begin{lemma}\label{ex-clifford}
Given two lists of signed Pauli operators \((P_1,\ldots,P_n)\) and \((Q_1,\ldots,Q_n)\). If
\begin{enumerate}
\item \(\forall i,j\leq n\), \(P_i,P_j\) commute \(\longleftrightarrow\) \(Q_i,Q_j\) commute,
\item \((P_1,\ldots,P_n)\) and \((Q_1,\ldots,Q_n)\) are both independent
\end{enumerate}
then there is a Clifford unitary \(U\) with \( UP_iU^\dag=Q_i\) for all $i$.
\end{lemma}

\begin{corollary}\label{SP-Circuit}
For any stabilizer projector \(A\) of rank \(2^m\) in an \(n\)-qubit  system, there is a Clifford circuit \(U_A\) such that 
\(U_A(\ketbra{0}{0}^{\otimes n-m} \otimes I) U_A^\dag=A\).
\end{corollary}

\section{SPD-Based Sampling Algorithm for Test Application}\label{sec-testapp}
Test application is the process of applying a test pattern obtained from test generation to the CUT and analyzing the output responses~\cite{wang2006vlsi}. More specifically, given a quantum test pattern $(\rho,M)$, the test application includes preparing the input quantum state \(\rho\), and performing measurement \(\{M,I-M\}\) on the output of CUT to obtain measurement results. By collecting measurement results from repeated experiments, we can estimate the value of \(\tr(M\CUT\rho \CUT^\dag)\), which represents the probability of the CUT passing the test. Consequently, if \(\tr(M\CUT\rho \CUT^\dag)>0.5\), we claim that the CUT is fault-free, otherwise faulty.

In this section, we present the stabilizer projector decomposition (SPD) and an SPD-based sampling algorithm for test application. Our algorithm adopts Clifford circuits as test equipment, which is robust against SPAM error and has better test equipment complexity.

\subsection{Stabilizer Projector Decomposition}
As mentioned in Subsection \ref{sec-stabilizer}, Clifford operations have reliable and efficient fault-tolerant implementation. We choose to use Clifford-only circuit to realize our test application. However, a test pattern $(\rho,M)$ as a solution of Equation (\ref{opt-cio}) is generally non-Clifford. To handle this problem, we introduce the stabilizer projector decomposition (SPD), which generalizes the notion of pseudo-mixtures for stabilizer states~\cite{howard2017application}.

\begin{definition}[Stabilizer Projector Decomposition, SPD]\label{def-SPD}
Let \(0 \sqsubseteq A\sqsubseteq I\). Then a stabilizer projector decomposition ({\SPD} for short) of \(A\) is a finite set \(\mathcal{A}=\{(a_i,A_i)\}_i\) such that $$A=\sum_i a_i A_i,$$ where 
  \(A_i\) are stabilizer projectors and \(a_i\) are real numbers. 
\end{definition}

We will use the calligraphic letters (\textit{e.g.} \(\mathcal{A},\mathcal{B}\)) to denote the SPDs and use \(\SPD(A)\) to denote the set of all possible {\SPD}s of operator \(A\). For convenience, given an SPD \(\mathcal{A}=\{(a_i,A_i)\}\), we use \(Q \otimes \mathcal{A}\) to denote the SPD \(\{(a_i,Q \otimes A_i)\}\),  where \(Q\) is a stabilizer projector. Similarly, we write \(U \mathcal{A}U^\dag\) for the SPD \(\{(a_i, U A_i U^\dag)\}\) and \(c\mathcal{A}\) for the SPD \(\{(c\cdot a_i,A_i)\}\) where \(U\) is a Clifford unitary, and \(c\) is a real number. We also define: $$\tr(\mathcal{A})=\tr(\sum_i a_i A_i).$$ 

In our SPD-based sampling algorithm, the following norms are needed. 

\begin{definition}[Norms of SPD]
The 1-norm \(\nu\) and weighted 1-norm \(\nu^*\) of  SPD \(\mathcal{A}=\{(a_i,A_i)\}\) are defined by:
\begin{equation}\small
\begin{gathered}
\nu(\mathcal{A})=\sum_i |a_i|\\
\nu^*(\mathcal{A})=\sum_i |a_i|\tr A_i
\end{gathered}
\end{equation}
\end{definition}

\subsection{SPD-based Sampling Algorithm}

Now we turn to the main problem in this section: estimate the expectation $\tr(M\CUT\rho \CUT^\dag)$ for a given test pattern $(\rho,M)$. Here we assume that the SPDs \(\mathcal{A},\mathcal{B}\) for $\rho,M$ are provided, respectively:
\begin{equation}\label{SPD_rhom}
\begin{gathered}
\rho=\sum_i a_i A_i,\\
M=\sum_i b_i B_i
\end{gathered}
\end{equation} but leave the problem of generating them to the next section. 

First, we can rewrite the expectation under estimation  in terms of the given SPDs:
\begin{equation}
\begin{split}
\tr(M\CUT\rho \CUT^\dag)&=\sum_{ij} a_ib_j \tr(B_j \CUT A_i \CUT^\dag)\\
&=\sum_{ij} a_i b_j \tr A_i \tr(B_j \CUT \frac{A_i}{\tr A_i} \CUT^\dag)
\end{split}
\end{equation}
For simplicity, we put: $$\tr_{ij}=\tr(B_j \CUT A_i/\tr A_i \CUT^\dag).$$

Since \(A_i,B_j\) are stabilizer projectors, by Corollary \ref{SP-Circuit}, we can implement the state preparation of \(A_i/\tr A_i\) and projective measurement of \(B_j\) using Clifford-only circuits. Then, we can sample from the Bernoulli distribution of parameter \(\tr_{ij}\), by applying the CUT on input state \(A_i/\tr A_i\) and sampling the result according to measurement \(B_j\).

To further estimate \(\sum_{ij} a_i b_j\tr A_i \tr_{ij}\), we use the quasi-probability distribution~\cite{howard2017application,pashayan2015estimating} \(\{a_ib_j\tr A_i\}_{ij}\) to form a true probability distribution \(\{p_{ij}\}_{ij}\),  where $$p_{ij}=|a_ib_j\tr A_i|/\sum_{k} |a_k\tr A_k|\sum_k|b_k|.$$ In each single trial, we sample \(i,j\) from this probability distribution \(p_{ij}\) and then sample a result \(m\in\{0,1\}\) from the Bernoulli distribution of parameter \(\tr_{ij}\). We do not record \(m\) but its corrected version:
\begin{equation}
\begin{split}
\hat{m}&=m\cdot \textup{sign}(a_ib_j)(\sum_k|a_k|\tr A_k)(\sum_k |b_k|)\\
&=m\cdot \textup{sign}(a_ib_j)\nu^*(\mathcal{A})\nu(\mathcal{B})
\end{split}
\end{equation}
where \(\textup{sign}(x)=1\) if \(x\geq 0\) and \(\textup{sign}(x)=-1\) otherwise. Note that \(\hat{m}\) is essentially sign-corrected by \(\textup{sign}(a_ib_j)\) and modulus-corrected by \(\nu^*(\mathcal{A})\nu(\mathcal{B})\). The expected value of random variable \(\hat{m}\) is
\begin{equation}
\begin{split}
E(\hat{m})&=\sum_{ij} p_{ij}\tr_{ij} \textup{sign}(a_ib_j)(\sum_k|a_k|\tr A_k)(\sum_k |b_k|)\\
&=\sum_{ij}a_ib_j\tr A_i \tr_{ij}=\tr(M\CUT\rho \CUT^\dag)
\end{split}
\end{equation}
This means \(\hat{m}\) can serve as an unbiased estimator of the target value \(\tr(M\CUT\rho \CUT^\dag)\).
By repeating this sampling process, we can estimate the result to any desired accuracy, where the number of repetition depends on the variance of \(\hat{m}\).
The details of this SPD-based sampling algorithm are summarized in Algorithm \ref{alg-smp}.

Then, we analyze the performance of Algorithm~\ref{alg-smp}. Note that \(\hat{m}\in [-\nu^*(\mathcal{A})\nu(\mathcal{B}),\nu^*(\mathcal{A})\nu(\mathcal{B})]\). By the Hoeffding's inequality, we know that 
\begin{equation}
\frac{2}{\delta^2}\ln(\frac{2}{\epsilon})\big[\nu^*(\mathcal{A})\nu(\mathcal{B})\big]^2
\end{equation}
samples are needed to estimate the result with error less than \(\delta\) and success probability exceeding \(1-\epsilon\). The above result is summarized in Theorem~\ref{alg-sdn}.

\begin{algorithm}[t!]
\small
\caption{SPD-based sampling algorithm}
\begin{algorithmic}[1]
  \Require SPD for \(\rho, M\) (see Equation (\ref{SPD_rhom})), Quantum CUT \(\CUT\), required precision \(\delta\) and success probability \(1-\epsilon\)
  \Ensure An estimation \(m\) of \(\tr(M \CUT\rho \CUT^\dag)\)
  \State Initialize the estimation result \(m\gets 0\).
  \State Set the iteration number \(T\gets \frac{2}{\delta^2}\ln(\frac{2}{\epsilon})[\nu^*(\mathcal{A})\nu(\mathcal{B})]^2\)
  \For {\(t=1,2,\ldots,T\)}
    \State Sample values \(i\) and \(j\) from distribution \(|a_i|\tr A_i/\nu^*(\mathcal{A})\) and \(|b_j|/\nu(\mathcal{B})\) respectively.
    \State Uniformly sample value \(l \in \{0,\ldots,\tr A_i-1\}\)
    \State Initialize the first \(n-\log \tr A_i\) qubits to \(\ket{0}\), and last \(\log \tr A_i\) qubits to \(\ket{l}\).
    \State Apply the Clifford circuit \(U_{A_i}\). \Comment Corollary \ref{SP-Circuit}
    \State Apply the CUT \(\CUT\).
    \State Apply the Clifford circuit \(U^\dag_{B_j}\). \Comment Corollary \ref{SP-Circuit}
    \State Measure the first \(n-\log \tr B_j\) quibts on computational basis.
    \If {The measurement result is \(\ket{0}\)}
      \State \(m\gets m+\textup{sign}(a_ib_j)\nu^*(\mathcal{A})\nu(\mathcal{B})\)
    \EndIf
  \EndFor
  \State $m\gets m/T$
  \State\Return \(m\)
\end{algorithmic}
\label{alg-smp}
\end{algorithm}

\begin{theorem}\label{alg-sdn}
Algorithm \ref{alg-smp} estimates \(\tr(M \CUT\rho \CUT^\dag)\) to additive error \(\delta\) with probability larger than \(1-\epsilon\), and has test time complexity (i.e., the number of experiments on the CUT):
\begin{equation}\label{eq-5262341}
O\bigg(\frac{\ln(1/\epsilon)}{\delta^2}\big[\nu^*(\mathcal{A})\nu(\mathcal{B})\big]^2\bigg).
\end{equation}
\end{theorem}

Note that the term \(\big[\nu^*(\mathcal{A})\nu(\mathcal{B})\big]^2\) corresponds to the additional overhead in test time complexity, introduced by our SPD-based sampling algorithm. As mentioned in Section~\ref{sec-5271640}, this can be viewed as a trade off between quantum resources (i.e., test equipment complexity) and classical resources (i.e., test time complexity). Specifically, our SPD-based sampling algorithm provides efficiency in test equipment complexity (to be justified in Table~\ref{table-526039}) and robustness in test application (due to the robustness of Clifford circuits), while it also requires more repeated experiment to obtain a valid test result (due to the uncertainty induced by the classical randomness). In the next section, we will focus on reducing the test time complexity, as well as the computational complexity of the test generation process. 

\section{SPD generation}\label{sec-spdgen}
In the previous section, an algorithm for robust test application was presented under the assumption that the SPDs of a test pattern \((\rho,M)\) are given. This section is devoted to solving the remaining problem, namely the SPD generation. Our aim is to optimize the test time complexity induced by the generated SPDs, while also reduce the computational complexity of the generation process.

As shown in Theorem \ref{alg-sdn}, we want to generate the SPDs that have as low \(\nu\) and \(\nu^*\) values as possible. In this sense, given a positive operator \(A\), one can find the \(\nu\)-optimal SPD by solving this convex optimization problem:
\begin{equation}\label{cvx-opt}
\begin{gathered}
\min_x\quad \|x\|_1\\
\textup{subject to}\quad \mathcal{S}x=b
\end{gathered}
\end{equation}
where \(\mathcal{S}_{ij}=\tr(P_iS_j)\), \(b_i=\tr(P_i A)\),  \(P_i\) is the \(i\)-th Pauli operator and \(S_j\) is the \(j\)-th stabilizer projector. An optimal solution \(x\) of problem (\ref{cvx-opt}) indicates a \(\nu\)-optimal SPD \(\{(x_i,S_i)\}\) for \(A\). For computational efficiency, we can discard those \((x_i,S_i)\) with \(x_i=0\). 
Similarly, we can obtain the \(\nu^*\)-optimal SPD by solving this convex optimization problem:
\begin{equation}\label{cvx-opt2}
\begin{gathered}
\min_x\quad \|w\circ x\|_1\\
\textup{subject to}\quad \mathcal{S}x=b
\end{gathered}
\end{equation}
where \(w_i=\tr(S_i)\) and \(\circ\) denotes the element-wise product. 

Now, our goal is to solve the SPD-optimization problems \eqref{cvx-opt2} and \eqref{cvx-opt} for the test pattern \(\rho\) and \(M\) obtained in \eqref{exp-rhoM}, respectively.
As noted in Section \ref{sec-pattern}, \(\rho\) and \(M\) can be obtained more efficiently through a fault propagation strategy. We adopt a similar idea here, embedding the SPD-optimization in the fault propagation process. That is, at the very beginning, we solve the SPD-optimization problems locally for the initial test pattern at the fault site, and use the resulting SPDs to represent the initial test pattern. Then, we forward (or backward) propagate the fault effect (or excitation signal) represented by the SPD to the primary output (or primary input) of the circuit, as with the D-algorithm. 
%We update the SPD recursively by re-solving the SPD-optimization problem \eqref{cvx-opt2} (or problem \eqref{cvx-opt}) in each propagation step.
More specifically, the propagation process is as follows:
\begin{itemize}
    \item Suppose at the current propagation step we have at hand a positive operator \(A\) representing the current outpu pattern \(M\) (or input pattern \(\rho\)), and also an SPD \(\mathcal{A}\) for \(A\). Suppose the next gate encountered is \(U\).
    \item We re-solve the optimization problem \eqref{cvx-opt2} (or problem \eqref{cvx-opt}) to find a new SPD for \(UAU^\dag\) (or \(U^\dag A U\)).
\end{itemize}
However, it is practically intractable to directly deal with these problems since the number of stabilizer projectors grows exponentially~\cite{gross2006hudson}. To alleviate this issue, several acceleration techniques are employed in each propagation step, such as the locality exploiting and sparsity exploiting.

\begin{algorithm}[t!]
\small
\caption{SPD generation}
\begin{algorithmic}[1]
  \Require Description of the fault-free circuit \(( U_1,\ldots,U_d)\), potential fault position \(i\) and corresponding fault model \(U'_i\).
  \Ensure SPD for \(\rho\) and \(M\), where \(\rho\) and \(M\) is the optimal input and measurement that can detect the fault.
  \State Let \((\ket{\psi}, \ket{\omega})\) be the test pattern for distinguishing \(U_i\) and \(U'_i\). \Comment see Section \ref{discriminate-gates}
  \State Let \(\mathcal{A}\) be the \(\nu^*\)-optimal SPD for \(\ketbra{\psi}{\psi}\) \Comment Problem (\ref{cvx-opt2})
  \State Let \(\mathcal{B}\) be the \(\nu\)-optimal SPD for \(\ketbra{\omega}{\omega}\) \Comment Problem (\ref{cvx-opt})
  \State \(\mathcal{A}\gets I\otimes \mathcal{A}\), \(\mathcal{B}\gets I\otimes \mathcal{B}\) \Comment Padding (choosing \(Q=I\) in Equation \eqref{exp-rhoM})
  \For {\(j=i-1\) down to \(1\)} \Comment Backward propagation
    \If {\(U_j\) is Clifford}
     	\State \(\mathcal{A}\gets U^\dag_j\mathcal{A} U_j\)
    \Else
        \State Let \(U_C=\text{LOC\_EXPL}(U_j^\dag, \mathcal{A})\) \Comment Algorithm \ref{alg-red}
        \State \(\mathcal{A} \gets\) the optimal SPD computed through diagram (\ref{diag-loc})
  	\EndIf
  \EndFor
  \State \(\mathcal{A}\gets \mathcal{A}/\tr(\mathcal{A})\)
  \For {\(j=i+1\) to \(d\)} \Comment Forward propagation
    \If {\(U_j\) is Clifford}
        \State \(\mathcal{B}\gets U_j\mathcal{B} U_j^\dag\)
    \Else
       	\State Let \(U_C=\text{LOC\_EXPL}(U_j, \mathcal{B})\) \Comment Algorithm \ref{alg-red}
       	\State \(\mathcal{B} \gets\) the optimal SPD computed through diagram (\ref{diag-loc})
    \EndIf
  \EndFor
  \State \Return \(\mathcal{A}, \mathcal{B}\)
\end{algorithmic}
\label{alg-spd}
\end{algorithm}

\subsection{Locality Exploiting}\label{sec-213336}
In this section, we introduce the locality exploiting technique in the process of SPD propagation. The idea is, \textbf{1)} using a Clifford unitary \(U_C\) to map the current SPD to a local form, \textbf{2)} solving problem (\ref{cvx-opt}) or (\ref{cvx-opt2}) in the ``non-trivial'' subsystem, and \textbf{3)} mapping back. The optimality is preserved thanks to the nice properties of \(\nu\) (and also \(\nu^*\) analogously), \textit{i.e.}, the Clifford invariance (C-Inv) and tensor invariance (T-Inv):
\begin{enumerate}
    \item \textbf{C-Inv}: For any Clifford unitary \(U_C\), if an SPD \(\mathcal{A}\) is \(\nu\)-optimal (\(\nu^*\)-optimal), then \(U_C\mathcal{A}U_C^\dag\) is also \(\nu\)-optimal (\(\nu^*\)-optimal).
    \item \textbf{T-Inv}: For any stabilizer projector \(S\), if an SPD \(\mathcal{A}\) is \(\nu\)-optimal (\(\nu^*\)-optimal), then \(S\otimes \mathcal{A}\) is also \(\nu\)-optimal (\(\nu^*\)-optimal).
\end{enumerate}
Specifically, suppose we need to find the \(\nu\)-optimal (or \(\nu^*\)-optimal) SPD for \(UA U^\dag\), and \(A\) is represented by an SPD \(\mathcal{A}\), then our strategy follows this diagram:
\begin{equation}\label{diag-loc}
\small
\begin{tikzcd}[column sep = 9em, row sep = 1.5em]
(U, \mathcal{A}) \arrow[d,"U_C"] \arrow[r,dashed, "\text{find optimal SPD}"] & U_C^\dag (U_\bot A_\bot U_\bot^\dag \otimes \mathcal{A}'') U_C\\
(U_\bot \otimes U', A_\bot \otimes \mathcal{A}') \arrow[d,"\text{localize}"] & U_\bot A_\bot U_\bot^\dag \otimes \mathcal{A}'' \arrow[u,"\text{C-Inv}"]\\
(U', \mathcal{A}') \arrow[r,"\text{solve problem (\ref{cvx-opt}) or (\ref{cvx-opt2}) for } U' A'U'^{\dag}"] & \mathcal{A}'' \arrow[u,"\text{T-Inv}"]
\end{tikzcd}
\end{equation}
where \(U_C\) is a Clifford unitary such that \(U \xrightarrow{U_C} U_\bot\otimes U'\) and \(\mathcal{A}\xrightarrow{U_C} A_\bot\otimes \mathcal{A}'\) by conjugation, where \(U_\bot\), \(A_\bot\) are Clifford unitary and stabilizer projector, which are the ``trivial'' parts in our computation.% and \(\bot=U_\bot A_\bot U_\bot^\dag\). 
The above procedures are summarized in Algorithm \ref{alg-spd}.

The remaining problem is how to find a good \(U_C\) that satisfies diagram (\ref{diag-loc}). We propose an algorithm (shown in Algorithm~\ref{alg-red}) for finding such \(U_C\) that reveals as much locality as possible. Similar to \cite{qassim2019clifford}, we assume that the elementary gate set under consideration is the set of Clifford + Pauli rotations (i.e., \(e^{i\theta P}\) where \(P\) is a Pauli operator and \(\theta \in [-\pi,\pi]\)). Note that this form covers some important and practical universal quantum gate sets, such as Clifford + T gate set, or more generally the Clifford + Z rotation gate set. 

We provide here the main intuition behind Algorithm \ref{alg-red} (for simplicity, here we only consider \(\mathcal{A}\), without \(U\)), and the detailed analyses are presented in the Appendix. First note that any stabilizer projector is equivalent to \(\ketbra{0}{0}\otimes I\) under Clifford conjugation. Thus, to extract a stabilizer projector \(A_\bot\) from the SPD \(\mathcal{A}\), Algorithm \ref{alg-red} works roughly as follows,
\begin{enumerate}
\item extract a rank-1 projector \(\ketbra{0}{0}\) (steps 3-8), by computing the intersection of all components in \(\mathcal{A}\),
\item extract a full rank projector \(I\) (steps 9-20), by computing the complement of the union of all components in \(\mathcal{A}\),
\item combine previous results (step 21).
\end{enumerate}
Finally, the algorithm can exploit a smaller ``non-trivial'' subsystem of size \((s+t)/2\) (see Algorithm \ref{alg-red} for details). Its performance is theoretically guaranteed by the following theorem (the proof is deferred to the Appendix):
\begin{theorem}\label{opt-local}
Algorithm \ref{alg-red} finds a Clifford unitary \(U_C\) that satisfies diagram (\ref{diag-loc}). Moreover, if either of the following two conditions is satisfied,
\begin{enumerate}
\item \(\mathcal{A}\) is \(\nu\)-optimal and has the smallest \(\nu^*\) among all other \(\nu\)-optimal {\SPD}s for \(A\)
\item \(\mathcal{A}\) is \(\nu^*\)-optimal and has the smallest \(\nu\) among all other \(\nu^*\)-optimal {\SPD}s for \(A\)
\end{enumerate}
then \(U_C\) is optimal (\textit{i.e.}, reveals the minimum ``non-trivial'' subsystem) among all Clifford unitaries that satisfy diagram (\ref{diag-loc}).
\end{theorem}
Theorem~\ref{opt-local} states that our locality exploiting algorithm can always leverage some locality (i.e., convert the problem into a subsystem), and under certain mild conditions, it can maximize locality exploitation (i.e., convert the problem into a subsystem with the smallest possible size). Both conditions in Theorem~\ref{opt-local} can be achieved by a strategy named multi-objective optimization, more details are shown in section \ref{sec-mul-obj-opt}. However, it should be noted that, in practice, the system size might be too large to solve the multi-objective optimization problem, so that the conditions in Theorem~\ref{opt-local} may not hold at all times. Even worse, there might be no locality to exploit in some cases. To alleviate this issue, we also propose other acceleration techniques such as sparsity exploiting (see Section~\ref{sec-5272319}) and Clifford channel decomposition (see Section~\ref{sec-5272320}).

\begin{algorithm}[t!]
\small
\caption{Locality exploiting, LOC\_EXPL\((U,\mathcal{A})\)}
\begin{algorithmic}[1]
  \Require Non-Clifford \(U\), SPD \(\mathcal{A}=\{(a_i, A_i)\}_{i=1}^m\).
  \Ensure Clifford \(U_C\) that satisfies diagram (\ref{diag-loc}).
  \State Suppose \(U=e^{i\theta P_U}\), \(P_U\) is a Pauli operator and \(|\theta| \leq \pi\)
  \State Let \(G_i\) be the group associated to stabilizer projector \(A_i\) 
  \State Let \(\hat{P}_1,\ldots,\hat{P}_l\) be the independent generators of \(G_1\cap\ldots \cap G_m\)
  \If {\(P_U\) anti-commutes with some \(\hat{P}_i\)}
  	\State Let \(P_U\) only anti-commutes with \(\hat{P}_l\) (exchange the subscripts of \(\hat{P}_i\) or set \(\hat{P}_i \gets \hat{P}_i \hat{P}_j\) if necessary)
  	\State \(l\gets l-1\)
  \EndIf
  \State Generate the Clifford unitary \(V_1\) that maps \(\hat{P}_1,\ldots,\hat{P}_l,P_U\) to \(Z_1,\ldots,Z_l,Z_{l+1}\) \Comment Lemma \ref{ex-clifford}
  \State For each \(V_1G_iV_1^\dag\), find a set of independent generators of the form \(Z_1,\ldots,Z_l,I_l\otimes P_{i,1},\ldots,I_l\otimes P_{i,l_i}\)
  \State Find a maximal independent set \(Q_1,\ldots,Q_s\) of \(Z_1, P_{1,1},\ldots,P_{1,l_1},\ldots,P_{m,1},\ldots,P_{m,l_m}\)
  \For {\(i\leq s,i<j\leq s\)}
    \If {\(Q_i,Q_j\) anti-commute}
      \For {\(k\leq s,k\neq i,j\)}
        \State If \(Q_k,Q_i\) anti-commute, then \(Q_k\gets \underline{Q_kQ_j}\)
        \State If \(Q_k,Q_j\) anti-commute, then \(Q_k\gets \underline{Q_kQ_i}\)
      \EndFor
    \EndIf
  \EndFor
  \State Exchange the subscripts of \(Q_1,\ldots,Q_s\) so that \(Q_1,\ldots,Q_t\) commute with all \(Q_i,i\leq s\) and \(Q_{t+2i-1}\) only anti-commutes with \(Q_{t+2i}\) for \(i=1,\ldots,(s-t)/2\)
  \State Generate a Clifford unitary \(V_2\) that maps \(Q_1,\ldots,Q_t\) to \(Z_1,\ldots,Z_t\) and maps \(Q_{t+1},\ldots,Q_s\) to \(Z_{t+1},X_{t+1},\ldots,Z_{(t+s)/2},X_{(t+s)/2}\) \Comment Lemma \ref{ex-clifford}
  \State \(V\gets V_2 V_1\)
  \State \(U_C:= V\)
  \State \Return \(U_C\)
\end{algorithmic}
\label{alg-red}
\end{algorithm}

\subsection{Implementation Details}\label{sec-131311}
The algorithm presented above still faces some  practical challenges in its implementation. In this subsection, we propose several important  techniques to meet these challenges and thus improve the feasibility of the algorithm.

\subsubsection{Multi-Objective Optimization}\label{sec-mul-obj-opt}
To meet the requirements of Theorem \ref{opt-local}, we need to obtain an SPD that is primarily \(\nu\)-optimal and secondarily \(\nu^*\)-optimal in the forward propagation (or primarily \(\nu^*\)-optimal and secondarily \(\nu\)-optimal in the backward propagation). Consider the case of forward propagation. Let \(\overline{\theta}\) be the optimum value of the optimization problem (\ref{cvx-opt}) and define the secondary optimization problem:
\begin{equation}\label{multi-obj-opt}
\begin{gathered}
\min_x\quad \|\omega \circ x\|_1\\
\textup{subject to}\quad \mathcal{S}x=b,\,\, \|x\|_1\leq \overline{\theta}
\end{gathered}
\end{equation}
Then, the solution of Problem (\ref{multi-obj-opt})  is primarily \(\nu\)-optimal and secondarily \(\nu^*\)-optimal. In practice, we use the perturbation method~\cite{mangasarian1979nonlinear} to compute the solution at once. Consider the following problem:
\begin{equation}\label{multi-obj-opt2}
\begin{gathered}
\min_x\quad \|x\|_1+\epsilon \|\omega\circ x\|_1\\
\textup{subject to}\quad \mathcal{S}x=b
\end{gathered}
\end{equation}
It can be shown that, for sufficiently small \(\epsilon\), the optimal solution \(x\) of Problem (\ref{multi-obj-opt2}) is also the optimal solution of Problem (\ref{multi-obj-opt}). For the case of backward propagation, we have a similar conclusion.

\subsubsection{Sparsity Exploiting}\label{sec-5272319}
Given an SPD \(\mathcal{A}=\{(a_i, A_i)\}_{i=1}^k\), we say that \(\mathcal{A}\) is \(k\)-sparse as it has \(k\) non-zero terms in the decomposition. Since \(\mathcal{A}\) is represented by \(k\) pairs of real number and stabilizer projector, any non-trivial operation on \(\mathcal{A}\) has time complexity at least \(\Omega(k)\). Besides, in the worst case, \(k\) may grow exponentially with the number of non-Clifford gates applied on this SPD. To accelerate our algorithm, we exploit the sparsity of \(\mathcal{A}\) using the iterative log heuristic method~\cite{lobo2007portfolio}. This method finds a sparse solution in the feasible set by finding the local optimal point of the logarithmic function \(\sum_i\log(\delta+|x_i|)\). As this objective function is not convex, a heuristic method is applied where we first replace the objective by its first-order approximation, then solve and re-iterate. 

For efficiency, we only use the stabilizer projectors that already exist in \(\mathcal{A}\) as the basis, instead of using all stabilizer projectors. We also require the solution is primarily \(\nu\)-optimal and secondarily \(\nu^*\)-optimal (in the forward propagation for example). Combining the perturbation method~\cite{mangasarian1979nonlinear} mentioned earlier, we obtain the following optimization problem:
\begin{equation}\label{sparse-opt}
\begin{gathered}
\min_x \quad \|x\|_1+\epsilon\|\omega\circ x\|_1+\epsilon^2 \|L\circ x\|_1\\
\textup{subject to}\quad \mathcal{S}x=b
\end{gathered}
\end{equation}
where \(\mathcal{S}_{ij}=\tr(P_i A_j)\), \(P_i\) is the \(i\)-th Pauli operator, \(A_j\) is the \(j\)-th stabilizer projector in SPD \(\mathcal{A}\), \(b_i=\tr(P_i A)\) and the weight vector \(L\) is re-adjusted in each iteration based on the rule:
\begin{equation}
L_i=1/(\gamma+|x_i|)
\end{equation}
where \(\gamma\) is a small threshold value and the iteration number is fixed to \(5\) for simplicity. Note that in problem (\ref{sparse-opt}), the primary, secondary and tertiary objectives are \(\nu\), \(\nu^*\) and sparsity respectively and this technique also applies to the backward propagation similarly. We use this technique to make the SPD sparse when a non-Clifford gate is applied in the forward or backward propagation.

\subsubsection{Clifford Channel Decomposition}\label{sec-5272320}
If the system size after exploiting locality and sparsity is still large, we can use the Clifford channel decomposition~\cite{bennink2017unbiased} to directly compute the sub-optimal SPD. For example, we use \(\boldsymbol{Z}(\theta)\) (bold \(Z\)) to denote the quantum channel corresponding to the coherent rotation \(\rho\mapsto Z(\theta) \rho Z(\theta)^\dag\), where \(Z(\theta)=\begin{bmatrix}1&0\\0&e^{i\theta}\end{bmatrix}\). Then \(\boldsymbol{Z}(\theta)\) can be decomposed as
\begin{equation}
\boldsymbol{Z}(\theta)=\frac{1+\cos\theta-\sin\theta}{2}\boldsymbol{I}+\frac{1-\cos\theta-\sin\theta}{2}\boldsymbol{Z}+\sin\theta\boldsymbol{S}
\end{equation}
where \(I\) denotes the identity channel, \(\boldsymbol{Z}\) and \(\boldsymbol{S}\) denote the channel \(\rho\mapsto Z\rho Z^\dag\) and channel \(\rho \mapsto S\rho S^\dag\),  respectively. The numerical results demonstrate that this decomposition is nearly optimal for \(0\leq \theta\leq\pi\), \textit{i.e.}, it has the nearly smallest 1-norm of coefficients among all Clifford channel decomposition (however, this does NOT imply the near-optimality of the resulted SPD). This decomposition can also be generalized to any Pauli rotation channel \(\boldsymbol{P}(\theta):= \rho \mapsto e^{-i\theta/2 P}\rho e^{i\theta/2 P}\), where \(P\) is a Pauli operation (with plus or minus sign) and \(0\leq \theta\leq \pi\). We can find a Clifford unitary \(U\) that maps \(P\) to \(Z_1\). Then \(\boldsymbol{P}(\theta)\) can be decomposed as:
\begin{equation*}
\begin{split}
\boldsymbol{P}(\theta)&=\rho\mapsto\ U^\dag e^{-i\theta/2 Z_1} U\rho U^\dag e^{i\theta/2 Z_1}U\\
&= \boldsymbol{U^\dag}\cdot\boldsymbol{Z}_1(\theta)\cdot\boldsymbol{U}\\
&=\frac{1+\cos\theta-\sin\theta}{2}\boldsymbol{I}+\frac{1-\cos\theta-\sin\theta}{2}\boldsymbol{U^\dag}\boldsymbol{Z}_1\boldsymbol{U}+\sin\theta\boldsymbol{U^\dag}\boldsymbol{S}\boldsymbol{U}
\end{split}
\end{equation*}
which is also nearly optimal for \(0\leq\theta\leq\pi\).
Using the Clifford channel decomposition, we can directly apply these three Clifford channels on SPD separately and then merge the resulted SPDs.

\section{Evaluation}\label{sec-eva}

\subsection{Experimental Settings}

The algorithms proposed in the previous sections are implemented on Python3, with the help of IBM Qiskit~\cite{qiskit} --- an open source SDK for quantum computation --- and CVXPY~\cite{diamond2016cvxpy} --- an open source Python-embedded modeling language for convex optimization problems. Note that we use the Clifford circuit synthesis algorithm implemented by Qiskit, which is based on the method in \cite{bravyi2021clifford} and yields lower circuit size (though non-optimal in general). 

To demonstrate the effectiveness and scalability of our method, we conduct the experiments on several commonly used quantum circuits of  different sizes, including Quantum Fourier Transform (QFT)~\cite{nielsen2002quantum}, Bernstein-Vazirani (BV)~\cite{bernstein1997quantum} and Quantum Volume (QV)~\cite{moll2018quantum}. The Quantum Fourier Transform is the quantum analogue of the discrete Fourier transform, which is used in many quantum algorithms such as Shor's factoring, quantum hidden subgroup algorithm, as well as the QFT-based adder. The Bernstein-Vazirani algorithm is a quantum algorithm that can efficiently learn the secret string \(s\) of a linear Boolean function \(f(x)=x\cdot s\). The Quantum Volume circuits are used to quantify the ability (such as the operation fidelity, connectivity and available gate sets) of near-term quantum computing systems. The detailed information such as circuit size, the depth and number of non-Clifford gates of our benchmark circuits is summarized in Table \ref{table-bench}.

\begin{table}[h]
\centering
\caption{Benchmark Circuits}
\renewcommand{\arraystretch}{1.5}
\setlength{\tabcolsep}{2mm}{
\begin{tabular}{ccccc}
\hline
 & Qubits & Size & Depth & Non-Clifford Gates \\
\hline
QFT\_3 & 3 & 18 & 14 & 9 \\
QFT\_5 & 5 & 55 & 30 & 30 \\
QFT\_10 & 10 & 235 & 70 & 135 \\
QV\_5 & 5 & 135 & 43 & 105 \\
QV\_7 & 7 & 201 & 43 & 156 \\
BV\_10 & 10 & 29 & 12 & 0 \\
BV\_100 & 100 & 299  & 102 & 0 \\
\hline
\end{tabular}}\label{table-bench}
\end{table}

In our experiments, we adopt the missing-gate fault model~\cite{paler2012detection,bera2017detection,hayes2004testing}, where the gate is assumed to be missing (\textit{i.e.}, the gate is replaced by the identity operation). However, any unitary fault is equivalently handled by our algorithms, with no additional overhead in general. An example is shown in Fig. \ref{fig-exp1}, where both the missing-gate fault and a general unitary fault (the replaced-by-\(R_X(\pi/3)\) fault in this example) are handled on QFT\_3 circuit.

The remaining part of this section is arranged as follows. In subsection \ref{exp-tpg}, we evaluate the proposed SPD-generation algorithm (Algorithm \ref{alg-spd}) for test pattern generation, on several benchmark circuits. In subsection \ref{exp-simulation}, we evaluate the proposed sampling algorithm (Algorithm \ref{alg-smp}) for test application, using the previously generated test pattern, and the overall performance of our method for the fault detection of quantum circuit is also reported.

\begin{figure}[H]
\centering
\subfloat[The 3-quibt QFT circuit under test]{\label{fig-152206}
\scalebox{0.9}{
\Qcircuit @C=0.8em @R=0.2em @!R { \\
	 	\nghost{} & \gate{\mathrm{H}} & \gate{\mathrm{R_Z}\,(\mathrm{\frac{\pi}{4}})} & \ctrl{1} & \qw & \ctrl{1} & \gate{\mathrm{R_Z}\,(\mathrm{\frac{\pi}{8}})} & \ctrl{2} & \qw & \ctrl{2} & \qw & \qw & \qw & \qw & \qw & \qw & \qw&\nghost{}\\
	 	\nghost{} & \gate{\mathrm{R_Z}\,(\mathrm{\frac{\pi}{4}})} & \qw & \targ & \gate{\mathrm{R_Z}\,(\mathrm{\frac{-\pi}{4}})} & \targ & \gate{\mathrm{H}} & \qw & \gate{\textcolor{red}{\mathrm{R_Z}\,(\mathrm{\frac{\pi}{4}})}} & \qw & \qw & \ctrl{1} & \qw & \ctrl{1} & \qw & \qw & \qw&\nghost{}\\
	 	\nghost{} & \gate{\mathrm{R_Z}\,(\mathrm{\frac{\pi}{8}})} & \qw & \qw & \qw & \qw & \qw & \targ & \gate{\mathrm{R_Z}\,(\mathrm{\frac{-\pi}{8}})} & \targ & \gate{\mathrm{R_Z}\,(\mathrm{\frac{\pi}{4}})} & \targ & \gate{\mathrm{R_Z}\,(\mathrm{\frac{-\pi}{4}})} & \targ & \gate{\mathrm{H}} & \qw & \qw&\nghost{}\\
\\ }}
}
\par
\vspace{2mm}
\tikz{\draw[-,black, dashed](-9.5,0) -- (5.0,0);}
\par
\vspace{0.5mm}
%%%%%%%%%%%%%%%%%%%%%%%%%%%%%%%%%%%%%%%%%%%%%%%%%%
\subfloat[$0.25,\, \langle Z_2\rangle$]{ %IZI
\scalebox{0.9}{
\Qcircuit @C=1.0em @R=1em @!R {
	 	\nghost{} & \lstick{\ket{0}}  & \qswap \qwx[1]  & \qw& \nghost{} \\
	 	\nghost{} & \lstick{}  & \qswap  & \qw & \nghost{}\\
	 	\nghost{} & \lstick{} & \qw  & \qw& \nghost{}\\
}}
}
\hspace{3mm}
\subfloat[$0.383,\, \langle Y_2\rangle$]{ %IYI
\scalebox{0.9}{
\Qcircuit @C=1.0em @R=0.65em @!R {
	 	\nghost{}& \qswap \qwx[1] & \gate{\mathrm{S}} & \gate{\mathrm{H}} & \gate{\mathrm{S}} & \gate{\mathrm{X}} & \rstick{\bra{0}} \qw & \nghost{}\\
	 	\nghost{}& \qswap & \qw & \qw & \qw & \qw & \rstick{} \qw & \nghost{}\\
	 	\nghost{}& \qw & \qw & \qw & \qw & \qw & \rstick{} \qw & \nghost{}\\
}}
}
\hspace{3mm}
\subfloat[$0.383,\, \langle X_2X_3\rangle$]{ %XXI
\scalebox{0.9}{
\Qcircuit @C=1.0em @R=0.65em @!R {
	 	\nghost{} & \qswap \qwx[1] & \qswap \qwx[2] & \ctrl{2} & \gate{\mathrm{H}} & \rstick{\bra{0}} \qw & \nghost{}\\
	 	\nghost{} & \qswap & \qw & \qw & \qw & \rstick{} \qw & \nghost{}\\
	 	\nghost{} & \qw & \qswap & \targ & \qw & \rstick{} \qw & \nghost{}\\
}}
}
\hspace{3mm}
\subfloat[$-0.153,\, \langle -X_3,X_2\rangle$]{ %-XII,IXI
\scalebox{0.9}{
\Qcircuit @C=1.0em @R=0.65em @!R {
	 	\nghost{} & \qswap \qwx[2] & \gate{\mathrm{H}} & \gate{\mathrm{X}} & \rstick{\bra{0}} \qw & \nghost{}\\
	 	\nghost{} & \qw & \gate{\mathrm{H}} & \qw & \rstick{\bra{0}} \qw & \nghost{}\\
	 	\nghost{} & \qswap & \qw & \qw & \rstick{} \qw & \nghost{}\\
}}
}
\par
\subfloat[$0.388,\, \langle -X_3,-X_2\rangle$]{ %-XII,-IXI
\scalebox{0.9}{
\Qcircuit @C=1.0em @R=0.65em @!R {
	 	\nghost{} & \qswap \qwx[2] & \gate{\mathrm{H}} & \gate{\mathrm{X}} & \rstick{\bra{0}} \qw & \nghost{}\\
	 	\nghost{} & \qw & \gate{\mathrm{H}} & \gate{\mathrm{X}} & \rstick{\bra{0}} \qw & \nghost{}\\
	 	\nghost{} & \qswap & \qw & \qw & \rstick{} \qw & \nghost{}\\
}}
}\hspace{3mm}
\subfloat[$0.388,\, \langle X_3, Y_2\rangle$]{ %XII,IYI
\scalebox{0.9}{
\Qcircuit @C=1.0em @R=0.65em @!R {
	 	\nghost{} & \qswap \qwx[2] & \gate{\mathrm{H}} & \qw & \qw & \qw & \rstick{\bra{0}} \qw & \nghost{}\\
	 	\nghost{} & \qw & \gate{\mathrm{S}} & \gate{\mathrm{H}} & \gate{\mathrm{S}} & \gate{\mathrm{X}} & \rstick{\bra{0}} \qw & \nghost{}\\
	 	\nghost{} & \qswap & \qw & \qw & \qw & \qw & \rstick{} \qw & \nghost{}\\
}}
}\hspace{3mm}
\subfloat[$-0.153,\, \langle X_3,-Y_2\rangle$]{ %XII,-IYI
\scalebox{0.9}{
\Qcircuit @C=1.0em @R=0.65em @!R {
	 	\nghost{} & \qswap \qwx[2] & \gate{\mathrm{H}} & \qw & \qw & \rstick{\bra{0}} \qw & \nghost{}\\
	 	\nghost{} & \qw & \gate{\mathrm{S}} & \gate{\mathrm{H}} & \gate{\mathrm{S}} & \rstick{\bra{0}} \qw & \nghost{}\\
	 	\nghost{} & \qswap & \qw & \qw & \qw & \rstick{} \qw & \nghost{}\\
}}
}
\par
\vspace{2mm}
\tikz{\draw[-,black, dashed](-9.5,0) -- (5.0,0);}
\par
\vspace{0.5mm}
%%%%%%%%%%%%%%%%%%%%%%%%%%%%%%%%%%%%%%%%%%%%%%%%%%%%%%%%%%%%%
\subfloat[$0.154,\, \langle Z_2\rangle$]{ %IZI
\scalebox{0.9}{
\Qcircuit @C=1.0em @R=1em @!R {
	 	\nghost{} & \lstick{\ket{0}} & \qswap \qwx[1]  & \qw& \nghost{} \\
	 	\nghost{} & \lstick{}  & \qswap  & \qw & \nghost{}\\
	 	\nghost{} & \lstick{} & \qw  & \qw& \nghost{}\\
}}
}
\hspace{3mm}
\subfloat[$0.005,\, \langle X_1X_2\rangle$]{ %IXX
\scalebox{0.9}{
\Qcircuit @C=1.0em @R=0.65em @!R {
	 	\nghost{} & \lstick{\ket{0}} & \gate{\mathrm{H}} & \qswap \qwx[1] & \targ & \qw & \nghost{}\\
	 	\nghost{} & \lstick{} & \qw & \qswap & \ctrl{-1} & \qw & \nghost{}\\
	 	\nghost{} & \lstick{} & \qw & \qw & \qw & \qw & \nghost{}\\
}}
}
\hspace{3mm}
\subfloat[$-0.034,\, \langle -X_1X_2\rangle$]{ %-IXX
\scalebox{0.9}{
\Qcircuit @C=1.0em @R=0.65em @!R {
	 	\nghost{} & \lstick{\ket{0}} & \gate{\mathrm{H}} & \qswap \qwx[1] & \targ & \qw & \qw &\nghost{}\\
	 	\nghost{} & \lstick{} & \qw & \qswap & \ctrl{-1} & \gate{\mathrm{Z}} &\qw & \nghost{}\\
	 	\nghost{} & \lstick{} & \qw & \qw & \qw & \qw &\qw & \nghost{}\\
}}
}
\hspace{3mm}
\subfloat[$-0.034,\, \langle Y_2\rangle$]{ %IYI
\scalebox{0.9}{
\Qcircuit @C=1.0em @R=0.65em @!R {
	 	\nghost{} & \lstick{\ket{0}} & \gate{\mathrm{S}} & \gate{\mathrm{H}} & \gate{\mathrm{S}} & \qswap \qwx[1] & \qw & \nghost{}\\
	 	\nghost{} & \lstick{} & \qw & \qw & \qw & \qswap & \qw & \nghost{}\\
	 	\nghost{} & \lstick{} & \qw & \qw & \qw & \qw & \qw & \nghost{}\\
}}
}
\par
\subfloat[$0.005,\, \langle -Y_2\rangle$]{ %-IYI
\scalebox{0.9}{
\Qcircuit @C=1.0em @R=0.65em @!R {
	 	\nghost{} & \lstick{\ket{0}} & \gate{\mathrm{S}} & \gate{\mathrm{H}} & \gate{\mathrm{S}} & \qswap \qwx[1] & \qw & \qw & \nghost{}\\
	 	\nghost{} & \lstick{} & \qw & \qw & \qw & \qswap & \gate{\mathrm{X}} & \qw & \nghost{}\\
	 	\nghost{} & \lstick{} & \qw & \qw & \qw & \qw & \qw & \qw & \nghost{}\\
}}
}
\hspace{3mm}
\subfloat[$0.153,\, \langle X_1,X_2\rangle$]{ %IIX,IXI
\scalebox{0.9}{
\Qcircuit @C=1.0em @R=0.65em @!R {
	 	\nghost{} & \lstick{\ket{0}} & \gate{\mathrm{H}} & \qw & \nghost{}\\
	 	\nghost{} & \lstick{\ket{0}} & \gate{\mathrm{H}} & \qw & \nghost{}\\
	 	\nghost{} & \lstick{} & \qw & \qw & \nghost{}\\
}}
}
\hspace{3mm}
\subfloat[$0.153,\, \langle -X_1,-Y_2\rangle$]{ %-IIX,-IYI
\scalebox{0.9}{
\Qcircuit @C=1.0em @R=0.65em @!R {
	 	\nghost{} & \lstick{\ket{0}} & \gate{\mathrm{H}} & \gate{\mathrm{Z}} & \qw & \qw & \qw & \nghost{}\\
	 	\nghost{} & \lstick{\ket{0}} & \gate{\mathrm{S}} & \gate{\mathrm{H}} & \gate{\mathrm{S}} & \gate{\mathrm{X}} & \qw & \nghost{}\\
	 	\nghost{} & \lstick{} & \qw & \qw & \qw & \qw & \qw & \nghost{}\\
}}
}
\par
\subfloat[$0.207,\, \langle Z_2\rangle$]{ %IZI
\scalebox{0.9}{
\Qcircuit @C=1.1em @R=1.0em @!R {
	 	\nghost{} & \qswap \qwx[1] & \rstick{\bra{0}} \qw & \nghost{}\\
	 	\nghost{} & \qswap & \rstick{} \qw & \nghost{}\\
	 	\nghost{} & \qw & \rstick{} \qw & \nghost{}\\
}}
}
\hspace{3mm}
\subfloat[$0.361,\, \langle Y_2\rangle$]{ %IYI
\scalebox{0.9}{
\Qcircuit @C=1.0em @R=0.65em @!R {
	 	\nghost{}  & \qswap \qwx[1] & \gate{\mathrm{S}} & \gate{\mathrm{H}} & \gate{\mathrm{S}} & \gate{\mathrm{X}} & \rstick{\bra{0}} \qw & \nghost{}\\
	 	\nghost{}  & \qswap & \qw & \qw & \qw & \qw & \rstick{} \qw & \nghost{}\\
	 	\nghost{}  & \qw & \qw & \qw & \qw & \qw & \rstick{} \qw & \nghost{}\\
}}
}
\hspace{3mm}
\subfloat[$-0.137,\, \langle -Y_2\rangle$]{ %-IYI
\scalebox{0.9}{
\Qcircuit @C=1.0em @R=0.65em @!R {
	 	\nghost{} & \qswap \qwx[1] & \gate{\mathrm{S}} & \gate{\mathrm{H}} & \gate{\mathrm{S}} & \rstick{\bra{0}} \qw & \nghost{}\\
	 	\nghost{} & \qswap & \qw & \qw & \qw & \rstick{} \qw & \nghost{}\\
	 	\nghost{} & \qw & \qw & \qw & \qw & \rstick{} \qw & \nghost{}\\
}}
}
\hspace{3mm}
\subfloat[$-0.137,\, \langle -X_2X_3\rangle$]{ %-XXI
\scalebox{0.9}{
\Qcircuit @C=1.0em @R=0.65em @!R {
	 	\nghost{} & \qswap \qwx[1] & \qswap \qwx[2] & \ctrl{2} & \gate{\mathrm{H}} & \gate{\mathrm{X}} & \rstick{\bra{0}} \qw & \nghost{}\\
	 	\nghost{} & \qswap & \qw & \qw & \qw & \qw & \rstick{} \qw & \nghost{}\\
	 	\nghost{} & \qw & \qswap & \targ & \qw & \qw & \rstick{} \qw & \nghost{}\\
}}
}
\par
\subfloat[$0.361,\, \langle X_2X_3\rangle$]{ %XXI
\scalebox{0.9}{
\Qcircuit @C=1.0em @R=0.65em @!R {
	 	\nghost{} & \qswap \qwx[1] & \qswap \qwx[2] & \ctrl{2} & \gate{\mathrm{H}} & \rstick{\bra{0}} \qw & \nghost{}\\
	 	\nghost{} & \qswap & \qw & \qw & \qw & \rstick{} \qw & \nghost{}\\
	 	\nghost{} & \qw & \qswap & \targ & \qw & \rstick{} \qw & \nghost{}\\
}}
}
\hspace{3mm}
\subfloat[$0.345,\, \langle -X_3,-X_2\rangle$]{ %-XII,-IXI
\scalebox{0.9}{
\Qcircuit @C=1.0em @R=0.65em @!R {
	 	\nghost{} & \qswap \qwx[2] & \gate{\mathrm{H}} & \gate{\mathrm{X}} & \rstick{\bra{0}} \qw & \nghost{}\\
	 	\nghost{} & \qw & \gate{\mathrm{H}} & \gate{\mathrm{X}} & \rstick{\bra{0}} \qw & \nghost{}\\
	 	\nghost{} & \qswap & \qw & \qw & \rstick{} \qw & \nghost{}\\
}}
}
\hspace{3mm}
\subfloat[$0.345,\, \langle X_3,Y_2\rangle$]{ %XII,IYI
\scalebox{0.9}{
\Qcircuit @C=1.0em @R=0.65em @!R {
	 	\nghost{} & \qswap \qwx[2] & \gate{\mathrm{H}} & \qw & \qw & \qw & \rstick{\bra{0}} \qw & \nghost{}\\
	 	\nghost{} & \qw & \gate{\mathrm{S}} & \gate{\mathrm{H}} & \gate{\mathrm{S}} & \gate{\mathrm{X}} & \rstick{\bra{0}} \qw & \nghost{}\\
	 	\nghost{} & \qswap & \qw & \qw & \qw & \qw & \rstick{} \qw & \nghost{}\\
}}
}
\vspace{-2mm}
\caption{An illustrative example of the quantum test patterns represented by SPD. Sub-figure (a) is the 3-qubit QFT circuit using Clifford + \(\exp(i\theta P)\) gates. The red colored \(R_Z(\pi/4)\) gate is suspected to be faulty, and the following test patterns are automatically generated on two different fault models: \textbf{1)} Sub-figures (b)-(h) represent the test patterns for detecting the missing-gate fault, where sub-figure (b) is the input pattern and sub-figures (c)-(h) are the measurement patterns. 
\textbf{2)} Sub-figures (i)-(v) represent the test patterns for detecting a more general unitary fault, \textit{i.e.}, replaced-by-\(R_X(\pi/3)\) fault, where sub-figures (i)-(o) are the input patterns and sub-figures (p)-(v) are the measurement patterns. 
The number and the Pauli operators in the sub-caption of each circuit represent the coefficient and corresponding stabilizer group in the SPD.}
\label{fig-exp1}
%\vspace{-2mm}
\end{figure}

\subsection{Test Generation}\label{exp-tpg}

In this section, we apply the proposed SPD generation algorithm on different benchmark circuits and evaluate the generated test patterns. First, an illustrative example is shown in Fig. \ref{fig-exp1}. In this example, the SPD generation algorithm is applied on the 3-qubit QFT circuit (see sub-figure (a)), where the middle red-colored \(R_z(\pi/4)\) gate is suspected to be faulty. Two different fault models are adopted and the corresponding test patterns are generated respectively. The sub-figures (b)-(h) are test patterns for the missing-gate fault, and the sub-figures (i)-(v) are test patterns for the replaced-by-\(R_X(\pi/3)\) fault. The number and signed Pauli operators in the caption of each sub-figures (excluding sub-figure (a)) represent the coefficient and corresponding stabilizer projector of the SPD. The test patterns that have input (output) labels are input (output) patterns. The input labels represent the specified input, and those qubits without input labels are in the maximally mixed state. The output labels represent the specified measurement output, and those qubits without output labels are simply discarded.

We then evaluate the performance of the proposed SPD generation algorithm in detail. Specifically, we are concerned with the following metrics: \textbf{1)} the product of the SPD 1-norm \(\nu^*(\mathcal{A}_\rho)\nu(\mathcal{B}_M)\), where \(\mathcal{A}_\rho\) denotes the generated SPD for test input \(\rho\) and \(\mathcal{B}_M\) denotes the generated SPD for test measurement \(M\), \textbf{2)} the sparsity of the generated SPD, \textbf{3)} the size of the quantum Clifford circuit that implements the test pattern, and  \textbf{4)} the depth of the quantum Clifford circuit that implements the test pattern. The first metric \(\nu^*(\mathcal{A}_\rho)\nu(\mathcal{B}_M)\) could be treated as a quantity that measures how ``non-Clifford'' the SPD is. It heavily impacts the iteration number needed in the sampling algorithm (see Algorithm \ref{alg-smp}). The second metric reflects the computational complexity required for any non-trivial operation on the SPD. The third and fourth metrics reflect the complexity of the Clifford circuit implementing the test patterns. 

\begin{figure}[h]
\centering
\subfloat[]{
  \includegraphics[width=0.45\linewidth]{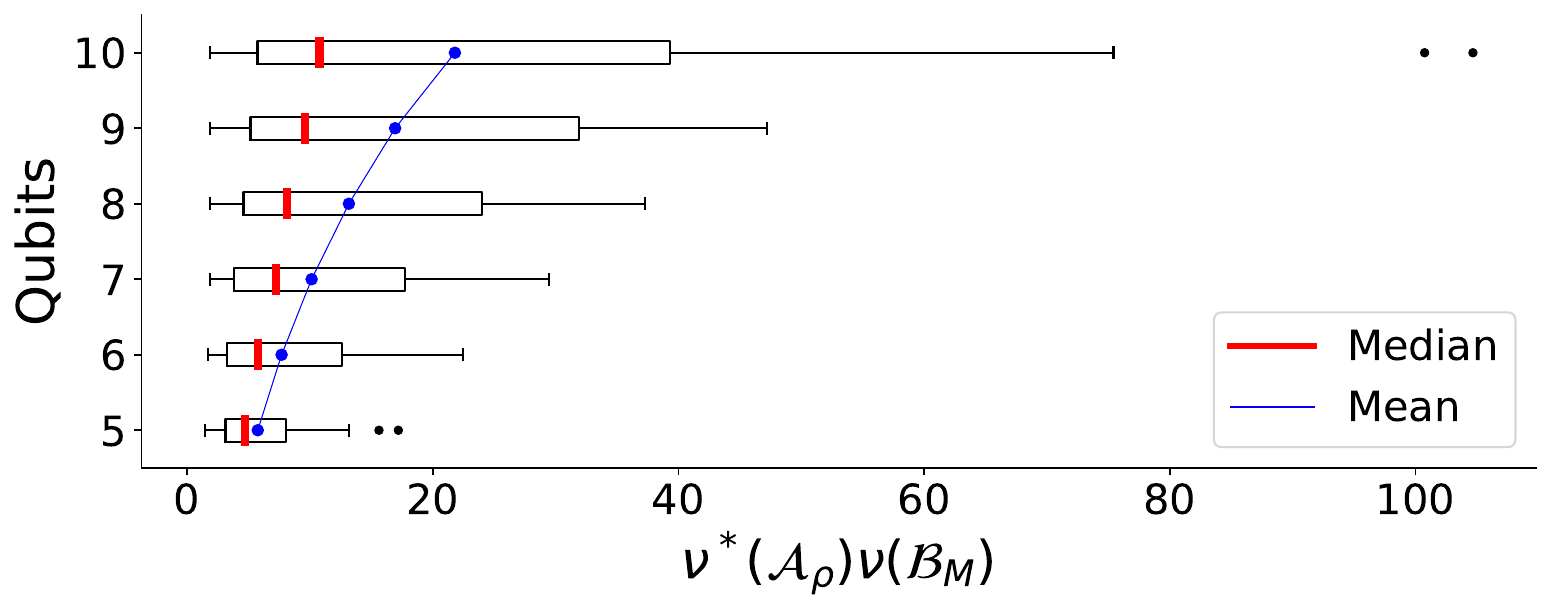}
  }
\hspace{2mm}
\subfloat[]{
  \includegraphics[width=0.45\linewidth]{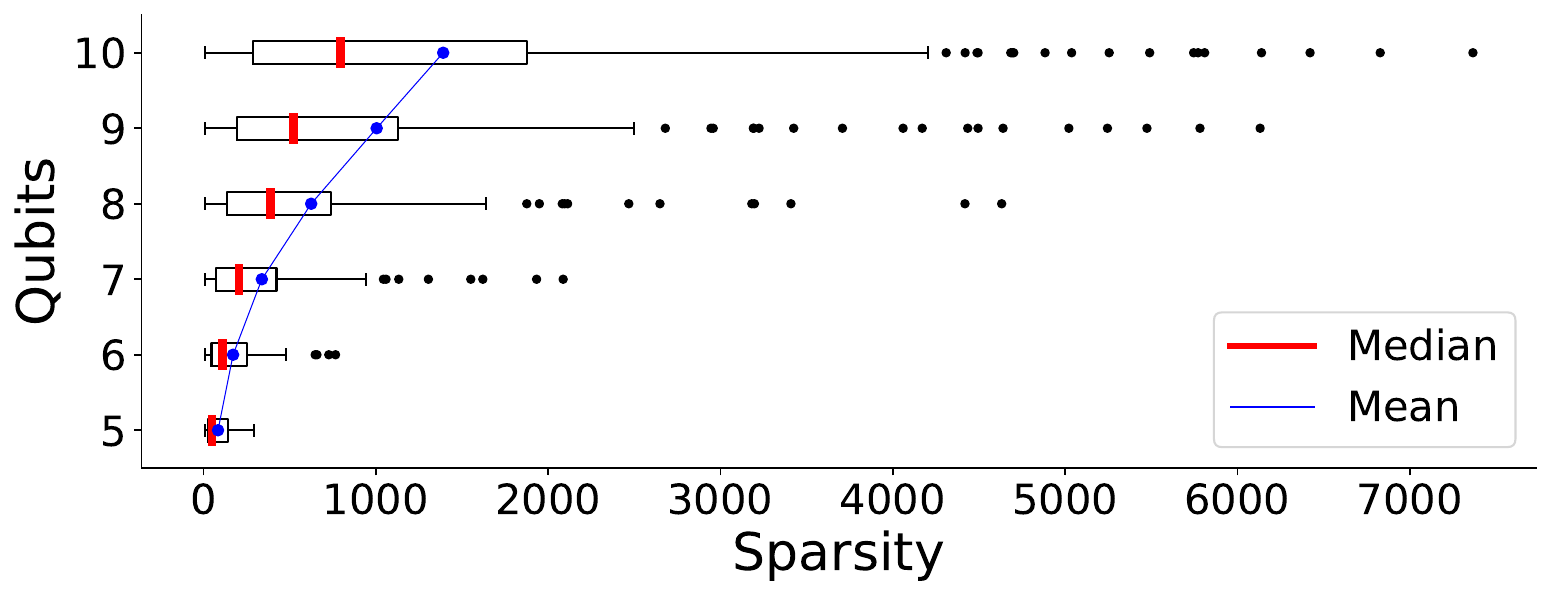}
  }
\par
\subfloat[]{
  \includegraphics[width=0.45\linewidth]{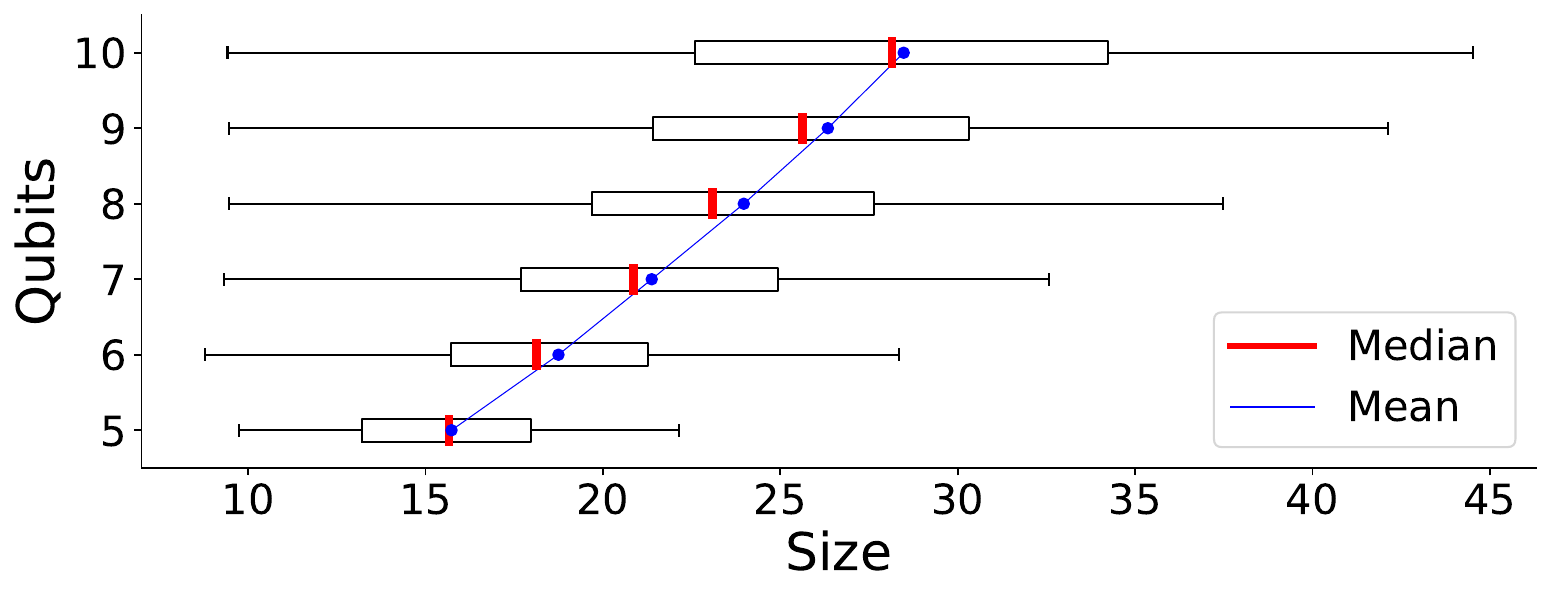}
  }
\hspace{2mm}
\subfloat[]{
  \includegraphics[width=0.45\linewidth]{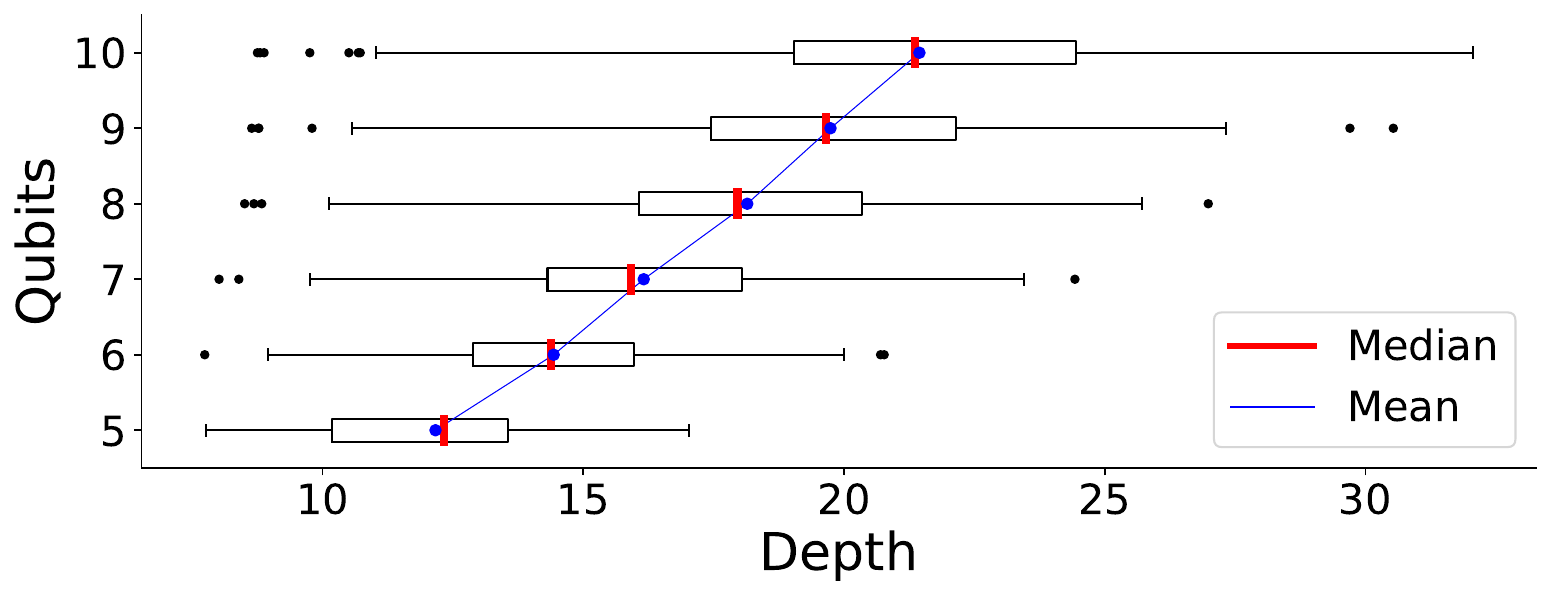}
  }
 \caption{Different metrics for the generated test patterns. The evaluation is based on QFT circuits and the circuit sizes range from 5 qubits to 10 qubits. The box extends from the first quartile to the third quartile of the data.}
\label{fig-spd}
\end{figure}

First, we conduct the experiments on QFT circuits with the sizes ranging from 5 qubits to 10 qubits. For each potential fault in the circuit, the corresponding test patterns (represented by SPD) are generated by our algorithm, and then evaluated with the aforementioned metrics. The results are shown in Fig. \ref{fig-spd}. It is worth noting that, though the maximum of \(\nu^*(\mathcal{A}_\rho) \nu(\mathcal{B}_M)\) grows rapidly, the median of \(\nu^*(\mathcal{A}_\rho) \nu(\mathcal{B}_M)\) grows slowly with the number of qubits. This demonstrates good scalability of the SPD in a considerable proportion of the cases. Similarly, the median of sparsity scales moderately compared to its maximum value, which demonstrates the computational efficiency of the SPD generation in the average sense. We believe that the SPD generation algorithm can be further accelerated using advanced parallel techniques on classical computer. We also observe that the averaged circuit size and depth of the test pattern scale almost linearly, which further ensures the practical feasibility of the proposed test application scheme.
Then, we also evaluate the QV and BV circuits, where the results are shown in Table \ref{table-spd}. Note that for the 7-qubit QV circuit, the metrics \(\nu^*(\mathcal{A}_\rho)\), \(\nu(\mathcal{B}_M)\) and the sparsity are rather higher than others. This is likely because QV is a  kind of randomized quantum circuits designed for quantifying the maximal capacity of a quantum device, so that it is much more complicated than most of the quantum circuits for normal uses. In contrast, the BV circuit contains only Clifford gates, so that our algorithm even scales well for the circuit size up to 100 qubits.

\begin{table}[t]
\centering
\caption{Averaged metric values of the generated test patterns on different benchmark circuits. The numbers in parentheses represent the optimal circuit size (depth) if we carefully choose the initial subspace \(Q\) (see Equation (\ref{exp-rhoM})).}
\renewcommand{\arraystretch}{1.5}
\setlength{\tabcolsep}{2mm}{
\begin{tabular}{cccccc}
\hline
 & \(\nu^*(\mathcal{A}_\rho)\) & \(\nu(\mathcal{B}_M)\) & Sparsity & Size & Depth \\
\hline
QFT\_5 & 1.698 & 3.381 & 84.0 & 15.7 & 12.2 \\
QFT\_6 & 1.956 & 4.056 & 171.5 & 18.7 & 14.4 \\
QFT\_7 & 2.272 & 4.764 & 338.4 & 21.4 & 16.2 \\
QFT\_8 & 2.653 & 5.484 & 624.8 & 24.0 & 18.1 \\
QFT\_9 & 3.112 & 6.222 & 1004.7 & 26.3 & 19.7 \\
QFT\_10 & 3.736 & 6.938 & 1390.8 & 28.5 & 21.4 \\
QV\_5 & 2.119 & 8.632 & 835.0 & 26.8 & 19.2 \\
QV\_7 & 3.209 & 15.006 & 5118.4 & 33.3 & 22.8 \\
BV\_10 & 1.493 & 1.479 & 10.3 & 18.9 & 15.9 \\
BV\_100 & 1.513  & 1.497 & 10.6 & 162.0 (37) & 145.2 (7) \\
\hline
\end{tabular}}\label{table-spd}
\end{table}

\subsection{Test Application and Fault Detection}\label{exp-simulation}

\begin{table}[t]
\centering
\caption{Comparison between different test application methods. ``Direct'' method uses the Clifford + Pauli-rotation gate set to implement the test patterns, while SPD-based method uses only the Clifford gate set. We report the averaged circuit sizes and depths, which reflects the test equipment complexity, and also the averaged \(\nu^*(\mathcal{A})\nu(\mathcal{B})\), which reflects the overhead in test time complexity.}
\renewcommand{\arraystretch}{1.5}
\setlength{\tabcolsep}{2mm}{
\begin{tabular}{@{\extracolsep{10pt}}cccccc@{}}
\hline
&\multicolumn{2}{c}{Direct}& \multicolumn{3}{c}{SPD-Based}\\
\cline{2-3}\cline{4-6} %\cmidrule{2-3}\cmidrule{4-5}
 & Size & Depth & Size & Depth & \(\nu^*(\mathcal{A})\nu(\mathcal{B})\) \\
\hline
QFT\_9 & 200.7 & 76.9 & 26.3 & 19.7 & 19.4 \\
QFT\_10 & 246.8 & 86.9  & 28.5 & 21.4 & 25.9 \\
QV\_5 & 141.6 & 54.6  & 26.8 & 19.2 &  18.3 \\
QV\_7 & 207.7 & 59.0 & 33.3 & 22.8 & 48.2 \\
BV\_10 & 36.8 & 16.7  & 18.9 & 15.9 & 2.2 \\
BV\_100 & 307.3  & 107.0 & 162.0 & 145.2 & 2.3 \\
\hline
\end{tabular}}\label{table-526039}
\end{table}

In this section, we first justify the prior mentioned quantum-classical resources trade off (see Fig.~\ref{fig-5261603} and related discussion). To see this, we compare the ``direct'' method with our SPD-based method, as shown in Table \ref{table-526039}. Both methods implement the same test patterns (i.e., \((\rho, M)\)), where the ``direct'' method uses Clifford + Pauli-rotation gate set, while the SPD-based method uses only Clifford gate set.  Specifically, we conduct the experiments on QFT, QV and BV circuits with different sizes. For each potential fault in the CUT, the corresponding test pattern is generated and evaluated, and the averaged results are given in Table \ref{table-526039}. Note that the sizes and depths of the circuits implementing the test patterns by our SPD-based method are significantly less than those by ``direct'' method. Moreover, our method only uses Clifford circuits. These demonstrate its efficiency and robustness.
%We also note that, \(\nu^*(\mathcal{A})\nu(\mathcal{B})\), which reflects the additional overhead in test time complexity, varies greatly among different circuits. For example, the overhead is small on the BV circuit even up to 100 qubits, while it becomes quite large on QV\_7 circuit. This is because this overhead largely depends on the ``non-Clifford'' components that the CUT contains, and will be further analyzed in the next experiments.}
%As mentioned in the previous section, \(\nu^*(\mathcal{A})\nu(\mathcal{B})\), which reflects the additional overhead in test time complexity, varies greatly among different circuits. Our method is more efficient on those circuits with less non-Clifford components.}

Then, we apply the test patterns (generated in Section~\ref{exp-tpg}) on benchmark circuits by simulating the sampling algorithm (Algorithm \ref{alg-smp}). Then, we show the overall results of our framework for quantum circuit fault detection. First, an illustrative example is shown in Fig. \ref{fig-distribution}. In this example, we simulate the sampling algorithm on the 3-qubit fault-free QFT circuit, using the test pattern for each gate in the circuit. The orange and green bars represent  the simulation results with different parameters, and the blue bar represents its expected value. The gate ID \(i\) indicates that the result is from applying the test patterns of \(i\)-th gate. First note that when setting \(\delta=0.1,\epsilon=0.1\), the estimation returned by  the sampling algorithm successfully approximates the ideal output. We found that when applying the test patterns of gate 6 and gate 7, the green bar is lower than \(0.5\). This means we may misclassify the CUT as faulty circuit when we set the classification threshold to be \(0.5\), resulting in a false positive alert. However, this is mainly because the gate itself is close to the identity gate to some extent. Thus it is intrinsically difficult to detect such fault under the SMGF (single missing gate fault) model. Except gate-6 and gate-7, there is no big difference between the results of using parameters \(\delta=0.1,\epsilon=0.1\) and \(\delta=0.3,\epsilon=0.3\) while the latter one requires fewer iterations. This means if the faults are sufficiently distinguishable, then the parameters \(\delta=0.3,\epsilon=0.3\) may be a practical choice that are both efficient and sound.

Then we evaluate our algorithms for the fault detection task. We first randomly selected \(k\) candidate fault-sites in the benchmark circuit such that the maximal discrimination probability is at least \(0.5+\tau\). This is because the discrimination probability of \(0.5+\tau\) indicates that the trace distance between fault-free and faulty output is at most \(2\tau\). Thus \(\tau\rightarrow 0\) indicates that the fault is negligible and inherently hard to detect. With the probability of \(0.5\), the CUT is set to be faulty, and otherwise the CUT is set to be fault-free. When the CUT is faulty, a fault-site is randomly selected from the \(k\) candidates and the gate at the fault-site is replaced by its faulty version. 
For each candidate fault-site \(i\), the corresponding test patterns are automatically generated and applied on the CUT by simulating the sampling algorithm (with parameters \(\epsilon,\delta\)) to obtain the estimation \(m_i\). The CUT is predicted to be faulty if \(\min\{m_1,\ldots,m_k\} \leq 0.5\), and otherwise is predicted to be fault-free. Note that we need to correctly obtain \(k\) estimations to ensure that the final prediction is correct. Since the parameter \(\epsilon\) bounds the success probability of our sampling algorithm from below, we set this parameter to \(\epsilon/k\) to ensure that the overall success probability for fault detection is at least \(1-\epsilon\). Note that the complexity of our sampling algorithm scales logarithmically with \(1/\epsilon\), thus the overhead is still acceptable. In this work, we set \(k=10,\tau=0.1,\delta=0.3,\epsilon=0.3\). For each benchmark circuit, we repeat the above fault detection procedure \(100\) times and the experimental results are shown in Table \ref{table-detect}. 

\begin{figure}[t]
\centering
\includegraphics[width=0.95\linewidth]{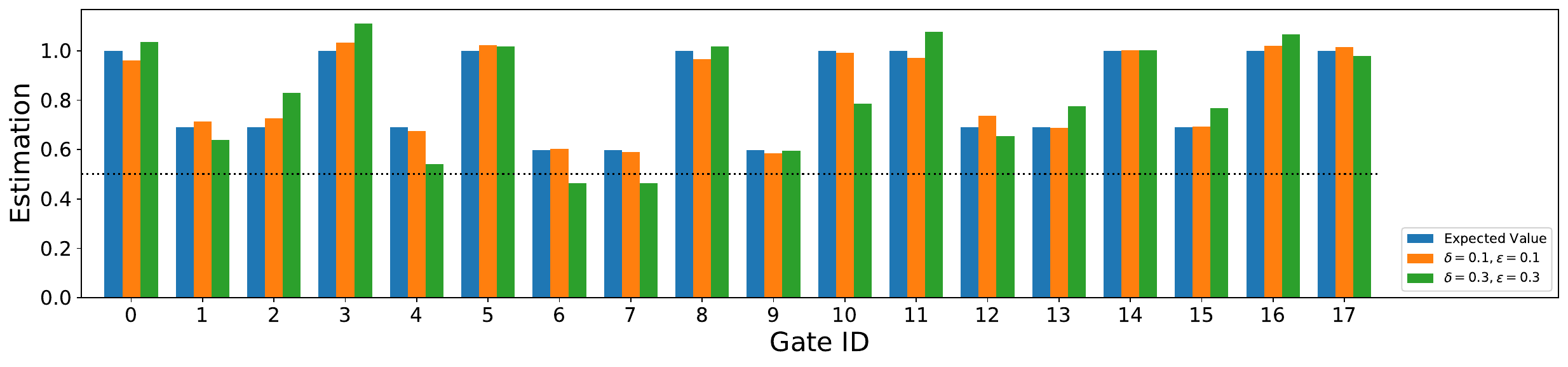}
 \caption{An illustrative example for the sampling algorithm (Algorithm \ref{alg-smp}) on 3-qubit QFT circuit. The orange and green bars represent the simulation results with different parameters, and the blue bar represents the expected output. The gate ID \(i\) means that we apply the test patterns for the \(i\)-th gate.}
\label{fig-distribution}
\end{figure}

\begin{table}[t]
\centering
\caption{Simulation results for quantum fault detection.}
\renewcommand{\arraystretch}{1.5}
\setlength{\tabcolsep}{2mm}{
\begin{tabular}{cccccccc}
\hline
 & TP & TN & FP & FN & Prec. & Rec. & Acc. \\
\hline
QV\_5 & 51 & 49 & 0 & 0 & 1.0 & 1.0 & 1.0 \\
BV\_10 & 47 & 53 & 0 & 0 & 1.0 & 1.0 & 1.0 \\
QFT\_5 & 46 & 47 & 7 & 0 & 0.87 & 1.0 & 0.93 \\
QFT\_10 & 44  & 47 & 9 & 0 & 0.83 & 1.0 & 0.91 \\
\hline
\end{tabular}}
\label{table-detect}
\end{table}

Our method achieves the highest recall (1.0) uniformly in all experiments, which means all the faults in the quantum circuits are detected with success probability \(1.0\). We also found that there are some false positive alerts in the experiments for QFT\_5 and QFT\_10 circuits, \textit{i.e.}, some fault-free circuits are misclassified as faulty circuits. However, the precision is still higher than \(0.83\) and the false positive rate (FPR) is lower than \(0.16\). Furthermore, since the recall is sufficiently close to \(1.0\), the false positive alerts can be effectively eliminated by repeating the test on the false positive samples and excluding the samples once they are classified as negative (\textit{i.e.}, fault-free). Finally, the overall performance of fault detection is much better than expected by the theoretical bounds (failure probability \(\epsilon=0.3\) and estimation deviation \(\delta=0.3\)), which indicates rather good performance of our method in practice.

\section{Conclusion}\label{sec-con}
In this paper, we propose a new ATPG framework for robust quantum circuit testing. The stabilizer projector decomposition (SPD) is introduced for representing the quantum test pattern, and an SPD generation algorithm is developed. Furthermore, the test application is reconstructed using Clifford-only circuits with classical randomness. To tackle the issue of exponentially growing size of the involved  optimization problem, we proposed several acceleration techniques that can exploit both locality and sparsity in SPD generation. It is  proved that the locality exploiting algorithm will return optimal result under some reasonable conditions. To demonstrate the effectiveness and efficiency, the overall framework is implemented and evaluated on several commonly used benchmark circuits.

For future research, we plan to extend our work in the following two directions: \subsubsection{Adaptive quantum ATPG}  The adaptive ATPG method~\cite{holst2009adaptive} for classical circuits utilizes the existing output responses from the CUT to guide the automatic generation of new test patterns, which makes the fault diagnosis much more efficient. In the quantum case, however, the output responses are quantum states, and thus we cannot easily make use of the information in the output responses due to the high cost of quantum state tomography. Further, the response analysis and pattern generation guidance are much more complicated than those for classical circuit. These make the development of adaptive quantum ATPG high nontrivial.
\subsubsection{Different Fault Models and Circuit  Models} In this paper, we assume that the faulty gate is still a unitary gate. This assumption does cover some important classes of faults in the real-world applications. But in many applications, a fault gate is often a non-unitary gate. Then techniques for discriminating a unitary gate and a general quantum operation (see e.g. \cite{gilchrist2005distance}) is vital in ATPG for this kind of faults. On the other hand, we only consider combinational quantum circuits in this paper. Recently, several new circuit models have been introduced in quantum computing and quantum information. One interesting model is dynamic quantum circuits (see e.g. \cite{corcoles2021exploiting, burgholzer2021towards}) in which quantum measurements can occur at the middle of the circuits and the measurement outcomes are used to conditionally control subsequent steps of the computation. A fault in these circuits may happen not only at a quantum gate but also at a measurement. Therefore, techniques for discriminating quantum measurements (see e.g. \cite{ji2006identification}) will be needed in ATPG for them. Another interesting model is  sequential quantum circuits~\cite{wang2021equivalence}. As we know, in the classical case, ATPG for combinational circuits and that for sequential ones are fundamentally different. The quantum case is similar, and  the approach presented in this paper cannot directly generalised to sequential quantum circuits.  

\section*{Acknowledgment}
This work was supported in part by the National Natural Science Foundation of China under Grant 61832015.

\bibliographystyle{ACM-Reference-Format}
\bibliography{egbib}

\appendix
\section{Appendix}

In this appendix, we provide the detailed proofs of  all lemmas, propositions and theorems omitted in the main body of the paper.  
\subsection{Minimal Norms of SPD}
\begin{definition}[Minimal SPD norms]
The minimal 1-norm $\xi$ and minimal weighted 1-norm $\xi^*$ of operator \(A\) are defined by
\begin{equation}\small
\begin{gathered}
\xi(A)=\min\bigg\{\nu(\mathcal{A})\;\bigg|\; \mathcal{A} \in \SPD(A) \bigg\}\\
\xi^*(A)=\min\bigg\{\nu^*(\mathcal{A})\;\bigg|\; \mathcal{A} \in \SPD(A) \bigg\}
\end{gathered}
\end{equation}

We say that \(\mathcal{A}\) is \(\nu\)-optimal (respectively,  \(\nu^*\)-optimal) for \(A\) if \(\mathcal{A}\in \SPD(A)\) and \(\nu(\mathcal{A})=\xi(A)\) (respectively,  \(\nu^*(\mathcal{A})=\xi^*(A)\)).
\end{definition}

The above definition is a generalization of the robustness of magic in quantum resource theory~\cite{howard2017application}. Some basic properties of \(\xi\) and \(\xi^*\) are given in the following:

\begin{proposition}\label{tens-inv}
Both \(\xi\) and \(\xi^*\) are invariant under applying Clifford unitary and tensoring stabilizer projector:
\begin{enumerate}
\item Clifford Invariance: for any Clifford unitary $U$, $$\xi(UAU^\dag)=\xi(A),\quad \xi^*(UAU^\dag)=\xi^*(A).$$  
\item Tensor Invariance: for any stabilizer projector $A$, $$\xi(A\otimes B)=\xi(B),\quad \xi^*(A\otimes B)=\tr A\cdot \xi^*(B).$$
\end{enumerate}
\end{proposition}
The Clifford invariance is obvious and the proof of tensor invariance is given in the next section.

\subsection{Proof of Proposition \ref{tens-inv} (Tensor Invariance)}

The proof of Proposition \ref{tens-inv} requires several technical lemmas:

\begin{lemma}\label{trace-out-0}
Suppose \(A\) is a stabilizer projector of bipartite system \(ST\), where \(S\) is a single qubit subsystem. Then
\begin{equation}
\bra{0}A\ket{0}=cB
\end{equation}
where \(B\) is a stabilizer projector of \(T\) and \(0\leq c\leq 1\). Furthermore, if \(c=1\), then \(A\sqsupseteq\ketbra{0}{0}\otimes B\).
\end{lemma}

\begin{proof}
Suppose \(G=\langle P_1,\ldots,P_l\rangle\) is the signed Pauli group associated to the stabilizer projector \(A\) where \(P_1,\ldots,P_l\) are independent generators of \(G\). Without loss of generality, we can assume that either \(Z_1\) commutes with all \(P_i\) or only anti-commutes with \(P_l\).

If \(Z_1\) commutes with all \(P_i\), then
\begin{equation}
\begin{split}
\ketbra{0}{0}\otimes cB&=\ketbra{0}{0}A\ketbra{0}{0}\\
&=\frac{I+Z_1}{2}A\frac{I+Z_1}{2}\\
&=\frac{I+Z_1}{2}\prod_{i\leq l} \frac{I+P_i}{2}\frac{I+Z_1}{2}\\
&=\prod_{i\leq l}\frac{I+P_i}{2}\frac{I+Z_1}{2}
\end{split}
\end{equation}
Since each \(P_i\) commutes with \(Z_1\), \(\langle P_1,\ldots,P_l,Z_1\rangle\) is a stabilizer group. Thus we can replace \(P_i\) by \(P_i Z_1\) (if necessary) to ensure that \(P_i\) is of the form \(I\otimes P'_i\). This means
\begin{equation}
\ketbra{0}{0}\otimes cB=\ketbra{0}{0}\otimes \prod_{i\leq l}\frac{I+P'_i}{2}
\end{equation}
Thus \(c=1\) and \(B\) is a stabilizer projector.

If \(Z_1\) anti-commutes with \(P_l\), then
\begin{equation}
\begin{split}
\ketbra{0}{0}\otimes cB&=\ketbra{0}{0}A\ketbra{0}{0}\\
&=\frac{I+Z_1}{2}A\frac{I+Z_1}{2}\\
&=\prod_{i\leq l-1}\frac{I+P_i}{2} \frac{I+Z_1}{2}\frac{I+P_l}{2}\frac{I+Z_1}{2}\\
&=\frac{1}{2}\prod_{i\leq l-1}\frac{I+P_i}{2} \frac{I+Z_1}{2}
\end{split}
\end{equation}
Similarly, we can assume \(P_i\) is of the form \(I\otimes P'_i\), which means
\begin{equation}
\ketbra{0}{0}\otimes cB= \ketbra{0}{0} \otimes \frac{1}{2}\prod_{i\leq l-1}\frac{I+P'_i}{2}
\end{equation}
Thus \(c=1/2\) and \(B\) is a stabilizer projector.

Furthermore, \(c=1\) indicates that all \(P_i\) commutes with \(Z_1\). So each \(P_i\) is either of the form \(I\otimes P'_i\) or of the form \(Z\otimes P'_i\). If \(P_i=I\otimes P'_i\) for all \(i\), then 
\begin{equation}
A=I\otimes \prod \frac{I+P'_i}{2}=I\otimes B\sqsupseteq \ketbra{0}{0}\otimes B
\end{equation}
Otherwise, there is some generator \(P_i=Z\otimes P'_i\). Without loss of generality, we can assume that only \(P_l\) has the form \(Z\otimes P'_l\). Then 
\begin{equation*}
\begin{split}
A&=(I\otimes \prod_{i<l} \frac{I+P'_i}{2}) \frac{I+Z\otimes P'_l}{2}\\
&=(I\otimes \prod_{i<l} \frac{I+P'_i}{2}) \big(\ketbra{0}{0}\otimes \frac{I+P'_l}{2}+\ketbra{1}{1}\otimes \frac{I-P'_l}{2}\big)\\
&=\ketbra{0}{0}\otimes  \prod_{i\leq l} \frac{I+P'_i}{2} +\ketbra{1}{1} \otimes \prod_{i<l} \frac{I+P'_i}{2}\frac{I-P'_l}{2}\\
&\sqsupseteq\ketbra{0}{0}\otimes B
\end{split}
\end{equation*}
\end{proof}

\begin{lemma}\label{trace-out-00}
Suppose \(A\) is a stabilizer projector of bipartite system \(ST\), where \(S\) is an \(n\)-qubit system, then
\begin{equation}
\bra{0}^{\otimes n}A \ket{0}^{\otimes n}=cB
\end{equation}
where \(B\) is a stabilizer projector of \(T\) and \(0\leq c \leq 1\). Furthermore, if \(c=1\), then \(A\sqsupseteq\ketbra{0}{0}^{\otimes n}\otimes B\)
\end{lemma}

\begin{proof}
By inductively apply Lemma \ref{trace-out-0}.
\end{proof}

\begin{lemma}\label{trace-out-I}
Suppose \(A\) is a stabilizer projector of bipartite system \(ST\), then
\begin{equation}
\tr_S(A)=cB
\end{equation}
where \(B\) is a stabilizer projector and \(0\leq c\leq \dim(S)\). Furthermore, if \(c=\dim(S)\), then \(A=I\otimes B\).
\end{lemma}
\begin{proof}
Suppose \(G=\{P_i\}_i\) is the stabilizer group associated to the stabilizer projector \(A\). Obviously, all the signed Pauli operators \(P_i\) that act trivially on system \(T\) (\textit{i.e.}, \(P_i=I\otimes P'_i\)) form a subgroup \(G'\). Then,
\begin{equation}
\begin{split}
\tr_S(A)&=\frac{1}{|G|}\sum_{P_i\in G} \tr_S(P_i)\\
&=\frac{1}{|G|}\sum_{P_i\in G'} \tr_S(P_i)\\ &=\frac{\dim(S)}{|G|}\sum_{P_i\in G'} P'_i\\
&=\frac{\dim(S)|G'|}{|G|}\frac{1}{|G'|}\sum_{P_i\in G'} P'_i\\
&=\frac{\dim(S)|G'|}{|G|} B
\end{split}
\end{equation}
So \(c=\dim(S)|G'|/|G|\leq \dim(S)\) and \(B\) is a stabilizer projector. Furthermore, if \(c=\dim(S)\), then \(|G'|=|G|\). This means \(G'=G\) and thus \(P_i=I\otimes P'_i\) for all \(P_i\in G\). So \(A=I\otimes B\).
\end{proof}

Now we are ready to give the proof of Proposition \ref{tens-inv}. For the clarity, let us split the proposition into the following two: 

\begin{proposition}[Tensor Invariance of \(\xi\)]
\begin{equation}
\xi(A\otimes B)=\xi(B)
\end{equation}
where \(A\) is a stabilizer projector.
\end{proposition}
\begin{proof}
We use \(S\) and \(T\) to denote the subsystems corresponding to the operators \(A\) and \(B\) respectively. By Corollary \ref{SP-Circuit}, there is a Clifford unitary \(U\) such that \(U AU^\dag=\ketbra{0}{0}^{\otimes m} \otimes I_n\), where \(m=\log\dim(S)-\log\tr(A)\) and \(n=\log\tr (A)\). Suppose \(\{(a_i,A_i)\}\) is an optimal SPD for \(\ketbra{0}{0}^{\otimes m} \otimes I_n\otimes B\), so
\begin{equation}
\ketbra{0}{0}^{\otimes m} \otimes I_n\otimes B= \sum_i a_iA_i
\end{equation}
By Lemma \ref{trace-out-00}
\begin{equation}
\begin{split}
I_n\otimes B=\sum_i a_i \bra{0}^{\otimes m}A_i\ket{0}^{\otimes m}
=c_1\sum_i a_i A'_i
\end{split}
\end{equation}
where \(0\leq c_1\leq 1\) and \(A'_i\) are stabilizer projectors.
By Lemma \ref{trace-out-I}
\begin{equation}
\begin{gathered}
2^n\cdot B=c_1\sum_i a_i \tr_n(A'_i)
=c_1c_2\sum_i a_i A''_i
\end{gathered}
\end{equation}
where \(0\leq c_2\leq 2^n\), \(\tr_n\) denotes the partial trace on first \(n\) quibts and \(A''_i\) are stabilizer projectors. Thus
\begin{equation}
B=c_1\cdot c_2/2^n\cdot \sum_i a_i A''_i
\end{equation}
This implies \(\xi(B)\leq \xi (\ketbra{0}{0}^{\otimes m} \otimes I_n\otimes B)= \xi(A\otimes B)\). But obviously we have \(\xi(A\otimes B)\leq \xi(B)\). Thus \(\xi(A\otimes B)=\xi(B)\).
\end{proof}

\begin{proposition}[Tensor Invariance of \(\xi^*\)]
\begin{equation}
\xi^*(A\otimes B)=\tr A\cdot\xi^*(B)
\end{equation}
where \(A\) is a stabilizer projector.
\end{proposition}

\begin{proof}
Since any stabilizer projector \(S\) can be expanded into the sum of \(\tr S\) stabilizer states, we can show that there is an stabilizer state decomposition
\begin{equation}
A\otimes B= \sum_i a_i \sigma_i
\end{equation}
where \(\sigma_i\) are stabilizer states and \(\sum_i |a_i|=\xi^*(A\otimes B)\). We use \(S\) to denote the subsystem of operator \(A\). By Lemma \ref{trace-out-I}
\begin{equation}
\begin{gathered}
\tr A\cdot B = \sum_i a_i \tr_S(\sigma_i)=\sum_i a_i c_i \sigma'_i
\end{gathered}
\end{equation}
where \(\sigma'_i\) are stabilizer projectors and \(c_i\tr(\sigma'_i)=\tr(\sigma_i)=1\). This means
\begin{equation}
B=\sum_i a_i \frac{1}{\tr A\cdot \tr(\sigma'_i)} \sigma'_i
\end{equation}
This implies \(\xi^*(B)\leq 1/\tr A\cdot \sum_i |a_i|=1/\tr A\cdot \xi^*(A\otimes B)\). But obviously we have \(\xi^*(A\otimes B)\leq \tr(A)\xi^*(B)\). Thus \(\xi^*(A\otimes B)=\tr A\cdot\xi^*(B)\)
\end{proof}

\subsection{Proof of Theorem \ref{opt-local} (Optimality of Algorithm \ref{alg-red})}
First, we show that the SPD is in a canonical form if one of the two conditions in Theorem \ref{opt-local} is satisfied. For clarity, we will prove the two cases separately.
\begin{proposition}\label{opt-form-tensor}
Suppose \(A\) is a stabilizer projector and \(B\) is a positive operator. If \(\mathcal{C}=\{(c_i,C_i)\}\in\SPD(A\otimes B)\) satisfies
\begin{enumerate}
\item \(\mathcal{C}\) is \(\nu\)-optimal, 
\item \(\mathcal{C}\) has the smallest \(\nu^*\) among all \(\nu\)-optimal SPDs of \(A\otimes B\)
\end{enumerate}
then each \(C_i\) has the form \(C_i=A\otimes C'_i\)
\end{proposition}

\begin{proof}
Suppose \(A\) is an \(n\)-qubit stabilizer projector and let \(m=\log\tr A\). Then, there is a Clifford unitary \(U\) such that \(UAU^\dag=\ketbra{0}{0}^{\otimes n-m} \otimes I_m\), where \(I_m\) is the identity operator on the \(m\)-qubit system. Note that \(U\) is Clifford, if \(\{(c_i,C_i)\}\) satisfies conditions 1) and 2) for \(A\otimes B\), then \(\{(c_i,U C_i U^\dag)\}\) also satisfies conditions 1) and 2) for \(UAU^\dag \otimes B\). This means we only need to prove this proposition with \(A=\ketbra{0}{0}^{n-m}\otimes I_m\) where \(m\geq 0\).

First, note that
\begin{equation}
\begin{split}
\sum_i c_i \bra{0}^{\otimes n-m} C_i \ket{0}^{\otimes n-m}&= \bra{0}^{\otimes n-m} A\otimes B\ket{0}^{\otimes n-m}\\
&=I_m \otimes B
\end{split}
\end{equation}
By Lemma \ref{trace-out-00}, \(c_i\bra{0}^{\otimes n-m}C_i\ket{0}^{\otimes n-m}=c'_i C'_i\) where \(|c'_i|\leq |c_i|\) and \(C'_i\) are also stabilizer projectors. This means 
\begin{equation}
\xi(I_m\otimes B)\leq \sum_i |c'_i|\leq \sum_i |c_i|=\xi(A\otimes B)
\end{equation}
However, by the tensor invariance of \(\xi\), \textit{i.e.}, $$\xi(A\otimes B)= \xi(\ketbra{0}{0}^{\otimes n-m}\otimes I_m\otimes B)=\xi(I_m\otimes B),$$ we can conclude that \(c_i=c'_i\). This means each \(C_i\) satisfies \(C_i\sqsupseteq \ketbra{0}{0}^{\otimes n-m}\otimes C'_i\) (see Lemma \ref{trace-out-00}). Note that \(\{(c_i,\ketbra{0}{0}^{\otimes n-m} \otimes C'_i)\}\) is also a \(\nu\)-optimal {\SPD} for \(A\otimes B\) since
\begin{equation}
\sum_i c_i \ketbra{0}{0}^{\otimes n-m} \otimes C'_i=\ketbra{0}{0}^{\otimes n-m}\otimes I_m\otimes B
\end{equation}
and has smaller \(\nu^*\), \textit{i.e.},
\begin{equation}
\sum_i |c_i| \tr( \ketbra{0}{0}^{\otimes n-m} \otimes C'_i)\leq \sum_i |c_i| \tr(C_i)
\end{equation}
So by condition 2), we conclude that \(C_i= \ketbra{0}{0}^{\otimes n-m} \otimes C'_i\).

Next, 
\begin{equation}
\begin{split}
2^m \cdot B&=\tr_m(I_m\otimes B)=\tr_m( \sum_i c_iC'_i)\\
&= \sum_i c_i \tr_m (C'_i)
\end{split}
\end{equation}
where \(\tr_m\) denotes the partial trace for the system of first \(m\) qubits. By Lemma \ref{trace-out-I}, \(c_i\tr_m(C'_i)=2^m c''_i C''_i\), where \(|c''_i|\leq |c_i|\) and \(C''_i\) are also stabilizer projectors. Thus \(B=\sum_i c''_i C''_i\) and $$\xi(B)\leq \sum_i |c''_i|\leq \sum_i |c_i|=\xi(A\otimes B).$$ Similarly, by the tensor invariance of \(\xi\), we can conclude that \(c''_i=c_i\). This means each \(C'_i\) has the form \(C'_i=I_m\otimes C''_i\) (see Lemma \ref{trace-out-I}).  
So, $$C_i=\ketbra{0}{0}^{\otimes n-m}\otimes I_m\otimes C''_i=A\otimes C''_i.$$
\end{proof}

\begin{proposition}\label{opt-form-tensor2}
Suppose \(A\) is a stabilizer projector and \(B\) is a positive operator. If \(\mathcal{C}=\{(c_i,C_i)\}\in\SPD(A\otimes B)\) satisfies
\begin{enumerate}
\item \(\mathcal{C}\) is \(\nu^*\)-optimal, 
\item \(\mathcal{C}\) has the smallest \(\nu\) among all \(\nu^*\)-optimal SPDs of \(A\otimes B\)
\end{enumerate}
then each \(C_i\) has the form \(C_i=A\otimes C'_i\)
\end{proposition}

\begin{proof}
Suppose \(A\) is an \(n\)-qubit stabilizer projector and let \(m=\log\tr A\). Similarly, we only need to prove this proposition with \(A=\ketbra{0}{0}^{n-m}\otimes I_m\) where \(m\geq 0\).

First, note that
\begin{equation}
\begin{split}
\sum_i c_i \bra{0}^{\otimes n-m} C_i \ket{0}^{\otimes n-m}&= \bra{0}^{\otimes n-m} A\otimes B\ket{0}^{\otimes n-m}\\
&=I_m \otimes B
\end{split}
\end{equation}
By Lemma \ref{trace-out-00}, \(c_i\bra{0}^{\otimes n-m} C_i\ket{0}^{\otimes n-m}=c'_i C'_i\) where \(C'_i\) are stabilizer projectors. Note that
\begin{equation}
\begin{split}
\tr C_i&=\sum_{j\in \{0,1\}^{n-m}}\tr \bra{j} C_i \ket{j}\\
&\geq \tr\bra{0}^{\otimes n-m}C_i\ket{0}^{\otimes n-m}\\
&=\frac{c'_i}{c_i}\tr C'_i
\end{split}
\end{equation}
which means \(|c_i|\tr C_i\geq |c'_i| \tr C'_i\). Then, we have
\begin{equation}
\xi^*(I_m\otimes B)\leq \sum_i |c'_i| \tr C'_i\leq \sum_i |c_i| \tr C_i=\xi^*(A\otimes B)
\end{equation}
However, by the tensor invariance of \(\xi^*\), \textit{i.e.}, $$\xi^*(A\otimes B)= \xi^*(\ketbra{0}{0}^{\otimes n-m}\otimes I_m\otimes B)=\xi^*	(I_m\otimes B).$$ We can conclude that \(c_i\tr C_i=c'_i \tr C'_i\) and so that \(\bra{j} C_i\ket{j}=0\) for all \(0^{n-m}\neq j\in\{0,1\}^{n-m}\). This means \(C_i=\ketbra{0}{0}^{\otimes n-m}\otimes C'_i\).

Next, we have:
\begin{equation}
\begin{split}
2^m \cdot B&=\tr_m(I_m\otimes B)=\tr_m( \sum_i c_iC'_i)\\
&= \sum_i c_i \tr_m (C'_i)
\end{split}
\end{equation}
where \(\tr_m\) denotes the partial trace for the system of first \(m\) qubits. By Lemma \ref{trace-out-I}, \(c_i\tr_m(C'_i)=2^m c''_i C''_i\), where \(|c''_i|\leq |c_i|\) and \(C''_i\) are also stabilizer projectors. Thus \(B=\sum_i c''_i C''_i\) and \(I_m\otimes B = \sum_i c''_i I_m\otimes C''_i\). Note that
\begin{equation}
\begin{gathered}
\sum_i |c''_i| 2^m \tr (C''_i)=\sum_i |c_i| \tr (C'_i)=\xi(I_m\otimes B)
\end{gathered}
\end{equation}
which means \(\{(c''_i,I_m\otimes C''_i)\}\) is a \(\nu^*\)-optimal SPD for \(I_m\otimes B\). It also has smaller \(\nu\), \textit{i.e.},
\begin{equation}
\sum_i |c''_i|\leq \sum_i|c_i|
\end{equation}
So by condition 2), we have \(|c''_i|=|c_i|\), and thus \(C'_i=I_m\otimes C''_i\) (see Lemma \ref{trace-out-I}).

So, \(C_i=\ketbra{0}{0}^{\otimes n-m}\otimes I_m\otimes C''_i=A\otimes C''_i\).
\end{proof}

Then, we are ready to prove Theorem \ref{opt-local}. We will first prove the correctness (\textit{i.e.}, the output of Algorithm \ref{alg-red} correctly reveal the locality) and then prove its optimality.
\begin{proof}[\textbf{Proof of Theorem \ref{opt-local}}]
\quad

\noindent\textbf{Correctness:} First we prove the correctness of Algorithm \ref{alg-red}. Note that \(V_1\) maps \(\hat{P}_1,\ldots,\hat{P}_l,P_U\) to \(Z_1,\ldots,Z_l,Z_{l+1}\). This means the group \(V_1G_i V_1^\dag\) can be generated by \(Z_1,\ldots, Z_l, \tilde{P}_{i,1}\otimes P_{i,1},\ldots, \tilde{P}_{i,l_i}\otimes P_{i,l_i}\), where \(\tilde{P}_{i,j}\) are Pauli operators on first \(l\) qubits. Note that \(\tilde{P}_{i,j}\) contains only \(I\) and \(Z\), as \(\tilde{P}_{i,j}\) commutes with all \(Z_1,\ldots,Z_l\). We can set \(\tilde{P}_{i,j}\) to \(I_l\) by multiplying several \(Z_k,k \leq l\) if necessary. Thus each \(V_1 G_i V_1^\dag\) has the generators \(Z_1,\ldots, Z_l, I_l\otimes P_{i,1},\ldots, I_l\otimes P_{i,l_i}\), which means \(V_1A_i V_1^\dag = \ketbra{0}{0}^{\otimes l}\otimes A'_i\) where $$A'_i=\prod_{j\leq l_i} \frac{I+P_{i,j}}{2}.$$ Besides, it holds that  $$V_1UV_1^\dag=V_1e^{i\theta P_U}V_1^\dag=e^{i\theta Z_{l+1}}=I_l\otimes e^{i\theta Z_1}.$$

Consider the subsystem that excludes the first \(l\) qubits. By steps 10-17 of Algorithm \ref{alg-red}, we can ensure that each \(Q_i\) at most anti-commutes with one \(Q_j\). Then, by exchanging the subscripts, we can make \(\{Q_i\}\) normal, \textit{i.e.}, \(\exists t\) such that \(Q_1,\ldots,Q_t\) commute with all \(Q_i\) for \(i\leq s\) and \(Q_{t+2i-1}\) only anti-commutes with \(Q_{t+2i}\) for \(i=1,\ldots,(s-t)/2\). Then, we can find a Clifford unitary \(V_2\) that maps \(\{Q_1,\ldots,Q_t\}\) to \(\{Z_1,\ldots,Z_t\}\) and maps \(\{Q_{t+1},\ldots,Q_s\}\) to \(\{Z_{t+1},X_{t+1},\ldots,Z_{(t+s)/2},X_{(t+s)/2}\}\) (see Lemma \ref{ex-clifford}). Since \(\{Q_1,\ldots,Q_s\}\) is the maximal independent set, \(V_2\) transforms all \(P_{i,j}\) and \(Z_1\) to the signed Pauli operators of the form \(I_{n'-(s+t)/2}\otimes P\) where \(P\) acts on a subsystem of size \((s+t)/2\). Thus, $$V_2 A'_iV_2^\dag=I_{n'-(s+t)/2}\otimes A''_i$$ for all \(i\) and $$V_2 U' V_2^\dag=I_{n'-(s+t)/2}\otimes U'',$$ where \(A''_i\) and \(U''\) are the stabilizer projectors and non-Clifford unitary in a subsystem of size \((t+s)/2\). Note that \(V=V_2V_1\) is then the required unitary, \textit{i.e.}, \begin{align*}V U V^\dag&=I_{n-(s+t)/2}\otimes U''\\
V A_i V^\dag&= \ketbra{0}{0}^{\otimes l} \otimes I_{n'-(s+t)/2}\otimes A''_i.\end{align*}

\noindent\textbf{Optimality:} Then we prove the optimality of Algorithm \ref{alg-red}. Suppose the optimal Clifford unitary is \(\overline{V}\) which satisfies \(\overline{V}A\overline{V}^\dag=B \otimes \overline{A}\) and \(\overline{V}U\overline{V}^\dag=W\otimes \overline{U}\), where \(B\) is a stabilizer projector, \(W\) is a Clifford unitary and \(\overline{A},\overline{U}\) act in a subsystem of size \(\overline{n}\). Without loss of generality, we can assume that \(B=\ketbra{0}{0}^{\otimes l}\otimes I_{n-\overline{n}-l}\) (by applying extra Clifford unitary on \(B\)). By conditions 1), 2) and Proposition \ref{opt-form-tensor} and \ref{opt-form-tensor2}, we know that for each \(A_i\), it can be written as follows:
\begin{equation}\label{eq-61}
A_i=\overline{V}^\dag (\ketbra{0}{0}^{\otimes l}\otimes I\otimes \overline{A}_i) \overline{V}    
\end{equation}
where \(\overline{A}_i\) are stabilizer projectors.
For convenience, we denote \(I_{n-\overline{n}-l}\otimes \overline{A}_i\) as \(\tilde{A}_i\).

Besides, \(W\otimes \overline{U}\) is of the form \(I_{n-\overline{n}}\otimes e^{i\theta \overline{P}_U}\). This is because \(e^{i\theta P_U}=U=\overline{V}^\dag (W\otimes \overline{U}) \overline{V}\) has two different eigenvalues (here we do not consider multiplicity of eigenvalues), which means \(W\) and \(\overline{U}\) at most have two different eigenvalues. As \(\overline{U}\) is a non-Clifford unitary, it has exactly two different eigenvalues \(\lambda_1,\lambda_2\). If \(W\) has exactly two different eigenvalues \(\mu_1,\mu_2\), then it can only be \(\lambda_1=-\lambda_2,\mu_1=-\mu_2\). Thus \(U=e^{i\theta P_U}\) has two eigenvalues \(\pm\lambda_1\mu_1\), which means \(\theta=\pm\pi/2\), contradicting that \(U\) is non-Clifford. So \(W\) has only one eigenvalue, which means \(W=I\) (up to a global phase). Thus \(\overline{U}=e^{i\theta \overline{P}_U}\) and \(P_U=\overline{V}^\dag (I\otimes \overline{P}_{U})\overline{V}\). For convenience, we denote \(I_{n-\overline{n}-l}\otimes \overline{P}_U\) as \(\tilde{P}_U\).

Define \(\overline{P}_i=\overline{V}^\dag Z_i \overline{V}\) for \(i=1,\ldots,l\). Note that \(Z_i\) commutes with \(\overline{P}_U\) for \(i\leq l\), so that \(\overline{P}_i\) commutes with \(P_U\) for \(i\leq l\). And since all \(Z_1,\ldots,Z_l\) are in the stabilizer group corresponding to \(\ketbra{0}{0}^\otimes l\otimes I\otimes \overline{A}_i\), by Equation (\ref{eq-61}), we have \(\overline{P}_1,\ldots,\overline{P}_l\in G_1\cap\ldots\cap G_m\equiv \hat{G}\) where \(G_i\) is the stabilizer group corresponding to \(A_i\). Since \(\overline{V}\) is optimal, \(\langle \overline{P}_1,\ldots,\overline{P}_l\rangle\) is the maximal sub-group of \(\hat{G}\) that commutes with \(P_U\) (otherwise there is some \(P_{l+1}\in\hat{G}\) such that \(\overline{P}_1,\ldots,\overline{P}_l,P_{l+1},P_U\) are independent and commute with each other. Then we can find another Clifford unitary that maps \(\overline{P}_1,\ldots,\overline{P}_l,P_{l+1}\) to \(Z_1,\ldots,Z_{l+1}\) and \(\overline{P}_U\) to \(Z_{l+2}\) to further shrink \(\overline{A}\) and \(\overline{U}\)). 

On the other hand, our algorithm exactly finds the unique maximal sub-group of \(\hat{G}\) that commutes with \(P_U\), \textit{i.e.}, \(\langle \hat{P}_1,\ldots,\hat{P}_l\rangle=\langle \overline{P}_1,\ldots,\overline{P}_l\rangle\). This is because if \(P_U\) commutes with \(\hat{G}\), then \(\hat{G}\) itself is the maximal sub-group. Otherwise suppose \(P_U\) does not commute with \(\hat{G}\). If there is another maximal subgroup \(G'\) that commutes with \(P_U\), let \(P'\in G'-\langle \overline{P}_1,\ldots,\overline{P}_l\rangle\), then 
\begin{equation}
|\langle \hat{P}_1,\ldots,\hat{P}_l,P'\rangle|=2^{l+1}=|\hat{G}|
\end{equation}
which means \(\langle \hat{P}_1,\ldots,\hat{P}_l,P'\rangle=\hat{G}\) and thus \(\hat{G}\) commutes with \(P_U\), contradiction. Thus the \(V_1\) found by our algorithm satisfies \(V_1 A_i V_1^\dag=\ketbra{0}{0}^{\otimes l}\otimes A'_i\) and \(V_1 P_U V_1^\dag=Z_{l+1}=I_l\otimes Z_1\). 

Note that \(V_1\) and \(\overline{V}\) are isomorphisms between \(Z_1,\ldots,Z_l\) and \(\langle \hat{P}_1,\ldots,\hat{P}_l\rangle=\langle \overline{P}_1,\ldots,\overline{P}_l\rangle\). Thus \(V_1\overline{V}^\dag\) is an automorphism on \(\langle Z_1,\ldots,Z_l\rangle\). 
Thus
\begin{equation}
\begin{split}
\ketbra{0}{0}^{\otimes l}\otimes A'_i &=V_1 A_i V_1^\dag\\
& = V_1 \overline{V}^\dag (\ketbra{0}{0}^{\otimes l}\otimes \tilde{A}_i)(V_1\overline{V}^\dag)^\dag\\
&=(\ketbra{0}{0}^{\otimes l}\otimes I) \big[V_1\overline{V}^\dag  (I\otimes \tilde{A}_i )(V_1\overline{V}^\dag)^\dag\big]
\end{split}
\end{equation}
So
\begin{equation}\label{AiAi}
A'_i=\bra{0}^{\otimes l} V_1\overline{V}^\dag (I\otimes \tilde{A}_i)(V_1\overline{V}^\dag)^\dag \ket{0}^{\otimes l}
\end{equation}
for all \(i\) and
\begin{equation}\label{Z1PU}
\begin{split}
I_l\otimes Z_1&=V_1P_U V_1^\dag\\
&=V_1\overline{V}^\dag (I\otimes \tilde{P}_U)(V_1\overline{V}^\dag)^\dag
\end{split}
\end{equation}

From now on, we will use \(\underline{P}\) to denote the unsigned Pauli operator by simply discarding the phase of \(P\). If the underline is applied on a set \(\underline{\{P,Q,\ldots\}}\), then it will denote \(\{\underline{P},\underline{Q},\ldots\}\).

Consider these two sets of Pauli operators:
\begin{equation}
\begin{split}
\overline{S}&=\textup{span}(\underline{\textup{SG}(\tilde{A}_1)},\ldots,\underline{\textup{SG}(\tilde{A}_m)},\underline{\tilde{P}_U})\\
S'&=\textup{span}(\underline{\textup{SG}(A'_1)},\ldots,\underline{\textup{SG}(A'_m)},Z_1)
\end{split}
\end{equation}
where \(\textup{SG}(A)\) denotes the stabilizer group of the stabilizer projector \(A\), ``span'' acts on the vector representations of the input and returns the set of (unsigned) Pauli operators corresponding to the spanned vector space. Then, we will show that \(\overline{S}\) and \(S'\) are isomorphic.

Define the mapping \(F: \overline{S}\rightarrow S'\) as
\begin{equation}
F(P)=\underline{\bra{0}^{\otimes l}V_1 \overline{V}^{\dag}(I\otimes P)(V_1\overline{V}^\dag)^\dag\ket{0}^{\otimes l}}
\end{equation}
\(F\) has the following properties:
\begin{enumerate}
\item \(F\) is bijective
\item \(P,Q\in \overline{S}\) commutes \(\Leftrightarrow\) \(F(P),F(Q)\) commutes.
\item \(\{P_1,\ldots,P_j\}\subseteq\overline{S}\) independent \(\Leftrightarrow\) \(\{F(P_1),\ldots,F(P_j)\}\) independent.
\end{enumerate}

\vspace{5mm}

\noindent \textbf{1)} \(F\) is bijective:

Note that
\begin{equation}
\begin{split}
&\quad\,\,\underline{\overline{V}^\dag \textup{span}(Z_1,\ldots,Z_l, I_l\otimes \overline{S})\overline{V}}\\
&=\textup{span}(\underline{\textup{SG}(A_1)},\ldots,\underline{\textup{SG}(A_m)},\underline{P_U})\\
&=\underline{V_1^\dag \textup{span}( Z_1,\ldots,Z_l, I_l\otimes S') V_1}
\end{split}
\end{equation}
This means \(\textup{span}( Z_1,\ldots,Z_l ,I_l\otimes\overline{S})\) and \(\textup{span}( Z_1,\ldots,Z_l ,I_l\otimes S')\) have the same size. And since 
$$\langle Z_1,\ldots,Z_l\rangle\cap (I_l\otimes\overline{S})=\langle Z_1,\ldots,Z_l\rangle\cap (I_l\otimes S')=\emptyset,$$ 
\(I_l\otimes\overline{S}\) and \(I_l\otimes S'\) have the same size, so do \(\overline{S}\) and \(S'\).
Equation (\ref{AiAi}) implies that \(F\) is a surjection of \(\underline{\textup{SG}(\tilde{A}_i)}\rightarrow \underline{\textup{SG}(A'_i)}\) for all \(i\). Thus \(F\) is surjection of \(\overline{S}\rightarrow S'\). And since \(|\overline{S}|=|S'|\), we conclude that \(F\) is a bijection of \(\overline{S}\rightarrow S'\).

\vspace{5mm}

\noindent\textbf{2)} \(P,Q\in \overline{S}\) commutes \(\Leftrightarrow\) \(F(P),F(Q)\) commutes:

Equation (\ref{AiAi}) implies that for all \(i\) and \(P\in \textup{SG}(A_i)\), \(\underline{V_1\overline{V}^\dag (I\otimes P) (V_1\overline{V}^\dag)^\dag}\) is of the form \(\perp_1\otimes\cdots\otimes \perp_l\otimes P'\), where \(\perp_j\in \{I,Z\}\) (since otherwise the RHS of Equation (\ref{AiAi}) equals \(c A\) where \(c<1\) and \(A\) is a stabilizer projector, see also the proof of Lemma \ref{trace-out-0} and Lemma \ref{trace-out-00}). Combined with Equation (\ref{Z1PU}), we conclude that for all \(P\in \overline{S}\), \(\underline{V_1\overline{V}^\dag (I\otimes P)(V_1\overline{V}^\dag)^\dag}\) is of the form \(\perp_1\otimes\cdots\otimes \perp_l\otimes P'\), where \(\perp_i\in\{I,Z\}\). 
We have
\begin{equation}
\forall P\in \overline{S},\quad \underline{V_1\overline{V}^\dag (I\otimes P) (V_1\overline{V}^\dag)}=\perp_P \otimes F(P)
\end{equation}
So
\begin{equation}
\begin{split}
&\quad\,\,P,Q \textup{ commutes }\\
&\Leftrightarrow \underline{V_1\overline{V}^\dag (I\otimes P) (V_1\overline{V}^\dag)^\dag}, \underline{V_1\overline{V}^\dag (I\otimes Q)(V_1\overline{V}^\dag)^\dag}\textup{ commutes}\\
&\Leftrightarrow \perp_P\otimes F(P),\perp_Q\otimes F(Q) \textup{ commutes} 
\end{split}
\end{equation}
Since \(\perp_P,\perp_Q\) always commutes,
\begin{equation}
P,Q \textup{ commutes }\Leftrightarrow\,\, F(P),F(Q) \textup{ commutes}
\end{equation}

\vspace{5mm}

\noindent\textbf{3)} \(\{P_1,\ldots,P_j\}\subseteq\overline{S}\) independent \(\Leftrightarrow\) \(\{F(P_1),\ldots,F(P_j)\}\) independent:

For \(\Leftarrow\) part, we have:
\begin{equation}
\begin{split}
&\quad\,\,\{F(P_i)\} \textup{ is independent}\\
&\Rightarrow \{\perp_{P_i}\otimes F(P_i)\}= \{\underline{V_1\overline{V}^\dag (I\otimes P_i)(V_1\overline{V}^\dag)^\dag}\} \textup{ is independent} \\
&\Rightarrow \{P_i\} \textup{ is independent}
\end{split}
\end{equation}

For the \(\Rightarrow\) part, we have:  
\begin{equation}
\begin{split}
\underline{(\perp_P\otimes F(P)) (\perp_Q\otimes F(Q))}
&= \underline{V_1\overline{V}^\dag PQ(V_1\overline{V}^\dag)^\dag}\\
&=\underline{\perp_P\perp_Q} \otimes F(\underline{PQ})
\end{split}
\end{equation}
Thus
\begin{equation}
\underline{F(P)F(Q)}=F(\underline{PQ})
\end{equation}
Suppose \(\{P_i\}\) is independent but \(\{F(P_i)\}\) is not. Without loss of generality, we can assume \(\underline{F(P_1)\cdots F(P_j)}=I\). Then \(F(\underline{P_1\cdots P_j})=I\). However, \(F(I)=I\) and \(\underline{P_1\cdots P_j}\neq I\). This contradicts the fact that \(F\) is bijective.

At this point, we can conclude that \(F\) satisfies the properties 1), 2) and 3). Thus \(F\) is an isomorphism between \(\overline{S}\) and \(S'\). 

Suppose \(\overline{s}\) is the dimension of \(\overline{S}\). By Lemma \ref{ext-normal}, there is a maximal independent \(\overline{t}\)-normal set (see the Definition \ref{t-normal}) \(\{P_1,\ldots,P_{\overline{s}}\}\subseteq\overline{S}\) for some \(\overline{t}\). Note that each \(P\in \overline{S}\) is of the form \(I_{n-\overline{n}-l}\otimes \perp_P\), where \(\perp_P\) is an \(\overline{n}\)-qubit Pauli operator. So there are \((\overline{s}+\overline{t})/2\) independent \(\overline{n}\)-qubit Pauli operators \(\perp_{P_1},\ldots,\perp_{P_{\overline{t}-1}},\perp_{P_{\overline{t}}},\perp_{P_{\overline{t}+2}},\perp_{P_{\overline{t}+4}},\ldots,\perp_{P_{\overline{s}}}\) that commute with each other. This implies that \(\overline{n}\geq (\overline{s}+\overline{t})/2\). 

By properties 2) and 3) of \(F\), \(\{F(P_1),\ldots,F(P_{\overline{s}})\}\) is also a maximal independent \(\overline{t}\)-normal set of \(S'\). Our algorithm finds another maximal independent \(t\)-normal set \(\{h_1,\ldots,h_s\}\) of \(S'\). Since both of them are maximal independent sets of \(S'\), we have \(s=\overline{s}\). By Lemma \ref{normalset}, we have \(t=\overline{t}\). So our algorithm find the unitary \(V_2\) that maps each \(A'_i\) and \(P'_U\) into a subsystem of  \((s+t)/2=(\overline{s}+\overline{t})/2\leq\overline{n}\). This proves the optimality of our algorithm.
\end{proof}

\begin{definition}\label{t-normal}
We say that a set of Pauli operators \(\{P_1,\ldots,P_s\}\) is \(t\)-normal if
\begin{enumerate}
\item \(P_1,\ldots,P_t\) commute with all \(P_i, 1\leq i\leq s\)
\item \(P_{t+2i-1}\) only anti-commutes with \(P_{t+2i}\) for \(1\leq i\leq (s-t)/2\)
\end{enumerate}
\end{definition}

\begin{lemma}\label{ext-normal}
Suppose \(S\) is a set of Pauli operators that satisfies \(\forall P,Q\in S,\,\underline{PQ}\in S\). Then, there is a independent subset \(\{P_1,\ldots,P_s\}\subseteq S\) such that
\begin{enumerate}
\item \(S=\textup{span}(P_1,\ldots,P_s)\),
\item \(\{P_1,\ldots,P_s\}\) is \(t\)-normal for some \(t\)
\end{enumerate}
We say that \(\{P_1,\ldots,P_s\}\) is a maximal independent \(t\)-normal set.
\end{lemma}
\begin{proof}
First find a maximal independent subset \(P_1,\ldots,P_s\subseteq S\), then use the similar manner as steps 10-17 of Algorithm \ref{alg-red} and exchange the subscripts if necessary.
\end{proof}

\begin{lemma}\label{normalset}
Suppose \(S\) is a set of Pauli operators that satisfies \(\forall P,Q\in S,\,\underline{PQ}\in S\). \(\{P_1,\ldots,P_s\}\) is a maximal independent \(t_1\)-normal set of \(S\) and \(\{Q_1,\ldots,Q_s\}\) is a maximal independent \(t_2\)-normal set of \(S\). Then \(t_1=t_2\).
\end{lemma}
\begin{proof}
Suppose \(t_1<t_2\). Since both \(\{P_i\}\) are maximal independent set, each \(Q_i\) can be represented as a combination of \(P_1,\ldots,P_s\). Note that there must be a \(Q\in\{Q_1,\ldots,Q_{t_2}\}\) that has components from \(P_{t_1+1},\ldots,P_s\) (otherwise \(\{Q_1,\ldots,Q_{t_2}\}\) can be represented by only \(\{P_1,\ldots,P_{t_1}\}\) and thus cannot be independent). Let this \(Q\) be \(Q_1\) and 
\[Q_1=\underline{Q'_1\prod_{i\in X} P_{i}}\]
where \(\emptyset \neq X\subseteq \{t_1+1,\ldots,s\}\) and \(Q'_1\) is composed of elements from \(\{P_1,\ldots,P_{t_1}\}\). Define the set \(Y\) to be
\begin{equation}
\begin{split}
Y=\{i| \exists j\in X, k\geq 1,\, (i,j)&=(t_1+2k-1,t_1+2k) \\
\vee\,\, (j,i)&=(t_1+2k-1,t_1+2k)\}
\end{split}
\end{equation}
Since \(X\) is non-empty, \(Y\) is also non-empty. Then, suppose each \(Q_i,i\geq 2\) has the form 
\[Q_i=\underline{Q'_i\prod_{j\in Y_i} P_j}\]
where \(Y_i\subseteq Y\) and \(Q'_i\) is composed of elements from \(\{P_k|k\notin Y \}\).

First note that \(|Y_i|\) must be even. Note that all \(Q_i,i\geq 2\) commute with \(Q_1\), which means \(\underline{\prod_{j\in Y_i} P_j}\) commutes with \(\underline{\prod_{j\in X}P_j}\). And whether \(\underline{\prod_{j\in Y_i} P_j}\) commutes with \(\underline{\prod_{j\in X}P_j}\) only depends on the parity of how many pairs of \((P_j,P_k), j\in X, k\in Y_i\) that are anti-commute. On the other hand, by the definition of \(Y\), we have 
\begin{equation}
\forall k\in Y, \exists!j\in X,\quad P_j \textup{ anti-commutes with } P_k
\end{equation} 
This implies that \(|Y_i|\) must be even to ensure that there are even pairs of \((P_j,P_k), j\in X, k\in Y_i\) that are anti-commute.

Besides, \(|X\cap Y|\) is always even. This is because, if \(z=t_1+2k\in X\cap Y\) then \(z-1=t_1+2k-1\in X\cap Y\), and if \(z=t_1+2k-1\in X\cap Y\) then \(z+1=t_1+2k\in X\cap Y\). Thus all elements in \(X\cap Y\) can be paired to \((2k-1,2k)\) for some \(k\geq 1\). 

Thus we conclude that for all \(i\), \(Q_i\) has exactly even components from \(P_Y=\{P_j|j\in Y\}\). Note that, \(\{Q_i\}\) is maximal independent, which means for any \(P\in P_Y\), \(P\) is representable by a combination of \(\{Q_i\}\). But any \(Q_i\) has even components from \(P_Y\), so that the combination of \(\{Q_i\}\) cannot represent \(P\). Contradiction!

Thus \(t_1\geq t_2\), and similarly we also have \(t_2\geq t_1\), which means \(t_1=t_2\).
\end{proof}

\end{document}